\documentclass[10pt,twocolumn,twoside]{IEEEtran}					%
\usepackage{etex}													%
\usepackage[english]{babel}											%
\usepackage[T1]{fontenc}											%
\usepackage{booktabs}												%
\usepackage{nccmath}												%
\usepackage[table]{xcolor}											%
\usepackage{amsmath,amssymb,amsfonts}								%
\usepackage[amsmath,thmmarks,framed,hyperref]{ntheorem}				%
\usepackage{acronym}												%
\usepackage{graphicx}												%
\DeclareGraphicsExtensions{.pdf,.png,.jpg,.mps}						%
\usepackage{soul}													%
\usepackage{url}													%
\usepackage{tikz}				
\usepackage{pgfplots}												%
\usetikzlibrary{calc}												%
\usepgflibrary{shapes.geometric}									%
\usepgfplotslibrary{groupplots}										%
\usepgfplotslibrary{statistics}										%
\pgfplotsset{compat=newest}											%
\usepackage[switch, pagewise, mathlines, running]{lineno}			%
\usepackage[color=white!90!black, textwidth=3cm]{todonotes}			%
\usepackage{algorithm}												%
\usepackage[noend]{algpseudocode}									%
\usepackage{enumerate}												%
\usepackage{enumitem}												%
\usepackage{cases}													%
\usepackage{bm}														%
\usepackage{ifthen}													%
\usepackage{mathtools}												%
\mathtoolsset{showonlyrefs=true}									%
\usepackage[nocompress]{cite}										%
\DeclareMathAlphabet{\mathpzc}{OT1}{pzc}{m}{it}						%
\newcounter{generalCounter}
\theoremstyle	{plain}
\newtheorem		{definition}	[generalCounter]	{Definition}
\newtheorem		{example}		[generalCounter]	{Example}
\newtheorem		{theorem}		[generalCounter]	{Theorem}

\newtheorem		{assumption}	[generalCounter]	{Assumption}

\newtheorem		{remark}		[generalCounter]	{Remark}

\newtheorem		{lemma}			[generalCounter]	{Lemma}

\def\TablesColumnsColor{black!4}
\newcolumntype
{g}
{
	>{\centering \columncolor{\TablesColumnsColor} \arraybackslash}
	p{0.15\textwidth}
	<{}
}
\newcolumntype
{w}
{
	>{\centering \arraybackslash}
	p{0.15\textwidth}
	<{}
}

\newcommand{\DiagonalMatrixOf}		[1]	{\mathrm{diag} \left( #1 \right)}

\newcommand{\TraceOf}				[1]	{\mathrm{tr} \left( #1 \right)}

\newcommand{\DefinedAs}			[0]	{\mathrel{\mathop:}=}
\newcommand{\IDefinedAs}		[0]	{=\mathrel{\mathop:}}

\newcommand{\BigOOf}			[1]	{O \left( #1 \right)}
\newcommand{\GaussianDistribution}					[2]	{\mathcal{N} \left( #1, #2 \right)}

\newcommand{\SpanOf}								[1]	{\textrm{span} \left\langle #1 \right\rangle}
\newcommand{\Reals}									[0]	{\mathbb{R}}

\newcommand{\Probability}			[0]	{\mathbb{P}}
\newcommand{\ProbabilityOf}			[1]	{\Probability \left[ #1 \right]}

\newcommand{\Expectation}					[0]	{\mathbb{E}}
\newcommand{\ExpectationOf}					[1]	{\Expectation \left[ #1 \right]}

\newcommand{\ExpectationOfGiven}			[2]	{\ExpectationOf{ #1 \; \left| \; #2 \right. }}

\newcommand{\Variance}				[0]	{\mathrm{var}}
\newcommand{\VarianceOf}			[1]	{\Variance \left( #1 \right)}
\newcommand{\VarianceOfGiven}		[2]	{\VarianceOf{ #1 \; \left| \; #2 \right. }}

\newcommand	{\Assumption}			[0]	{Assumption}

\newcommand	{\Definition}			[0]	{Definition}

\newcommand	{\Figure}				[0]	{Figure}

\newcommand	{\Lemma}				[0]	{Lemma}

\newcommand	{\Section}				[0]	{Section}

\newcommand	{\Table}				[0]	{Table}

\newcommand	{\Theorem}				[0]	{Theorem}
\acrodef{BLUE}		[BLUE]			{Best Linear Unbiased Estimator}
\acrodef{CE}		[CE]			{Classification Error}
\acrodef{GP}		[GP]			{Gaussian Process}
\acrodef{KL}		[KL]			{Karhunen-Lo\`{e}ve}
\acrodef{LTI}		[LTI]			{Linear Time Invariant}
\acrodef{MAP}		[MAP]			{Maximum A Posteriori}
\acrodef{MVUE}		[MVUE]			{Minimum Variance Unbiased Estimator}
\acrodef{MMSE}		[MMSE]			{Minimum Mean Square Error}
\acrodef{MSE}		[MSE]			{Mean Square Error}
\acrodef{ML}		[ML]			{Maximum Likelihood}
\acrodef{LMMSE}		[LMMSE]			{Linear Minimum Mean Square Error}
\acrodef{LS}		[LS]			{Least Squares}
\acrodef{WSN}		[WSN]			{Wireless Sensor Network}
\acrodef{NCS}		[NCS]			{Networked Control System}
\acrodef{RKHS}		[RKHS]			{Reproducing Kernel Hilbert Space}
\acrodef{RHS}		[RHS]			{Right Hand Side}
\acrodef{LHS}		[LHS]			{Left Hand Side}
\acrodef{RN}		[RN]			{Regularization Network}
\acrodef{SURE}		[SURE]			{Stein's unbiased risk estimate}
\acrodef{SVD}		[SVD]			{Singular Values Decomposition}
\acrodef{RSS}		[RSS]			{Residual sum of squares}
\def\MAP{\textrm{MAP}}

\begin{document}
\title{Distributed multi-agent Gaussian regression via finite-dimensional approximations}
\author
{
	Gianluigi Pillonetto, Luca Schenato, Damiano Varagnolo
	\thanks
	{
		G.\ Pillonetto and L.\ Schenato are with the Department of Information Engineering, University of Padova, Padova, Italy. D.\ Varagnolo is with the Department of Computer Science, Electrical and Space Engineering, Lule{\aa} University of Technology, Lule{\aa}, Sweden. Emails: {\tt giapi@dei.unipd.it | schenato@dei.unipd.it | damvar@ltu.se}. 
	}
	\thanks
	{
		The research leading to these results has received funding from the Swedish research council Norrbottens Forskningsr{\aa}d.
	}
}
\date{}
\maketitle
\IEEEoverridecommandlockouts
\IEEEpeerreviewmaketitle
\begin{abstract}
	We consider the problem of distributedly estimating Gaussian processes in multi-agent frameworks. Each agent collects few measurements and aims to collaboratively reconstruct a common estimate based on all data. Agents are assumed with limited computational and communication capabilities and to gather $M$ noisy measurements in total on input locations independently drawn from a known common probability density. The optimal solution would require agents to exchange all the $M$ input locations and measurements and then invert an $M \times M$ matrix, a non-scalable task. Differently, we propose two suboptimal approaches using the first $E$ orthonormal eigenfunctions obtained from the \ac{KL} expansion of the chosen kernel, where typically $E\ll M$. The benefits are that the computation and communication complexities scale with $E$ and not with $M$, and computing the required statistics can be performed via standard average consensus algorithms. We obtain probabilistic non-asymptotic bounds that determine a priori the desired level of estimation accuracy, and new distributed strategies relying on \ac{SURE} paradigms for tuning the regularization parameters and applicable to generic basis functions (thus not necessarily kernel eigenfunctions) and that can again be implemented via average consensus. The proposed estimators and bounds are finally tested on both synthetic and real field data.
\end{abstract}

\section{Introduction}

Many modern engineering problems involve networks containing a large number of agents which have to cooperate to obtain a common goal. Several of these tasks can be seen as problems of function estimation from sparse and noisy data, a central issue in the machine learning field~\cite{Poggio90,Cucker01}. Examples include the determination of the wind speed and direction field in a wind farm from local measurements of the turbines~\cite{lei2009review}, the reconstruction of the temperature field in a datacenter from local measurements at each server~\cite{parolini2012cyber}, and weather forecasts \cite{Gelfand:05,Datta:16}. Traditional centralized machine-learning estimation approaches are computationally non-scalable when the network is large. Moreover parallelization of computation using client-server architectures, which can alleviate this problem, might not be feasible. This happens, e.g., in applications where communication is peer-to-peer, as in wireless sensor networks or multi-agent robotics, and where each agent is expected to have a common copy of the global estimate. In these cases, fully distributed cooperation approaches are ought~\cite{YunfeiXu2012}.
\vspace{-2mm}
\subsection{State-of-the-art}

This paper considers a distributed nonparametric Gaussian regression approach. In this context, the unknown map is modeled as a zero-mean Gaussian process whose covariance (also called kernel in the machine learning literature) has to embed expected properties like smoothness~\cite{Scholkopf01b,rasmussen_williams__2006__gaussian_processes_for_machine_learning}. Other approaches to function estimation could be also adopted, e.g., sparse regression based on the $\ell_1$ norm, automatic relevance determination or the elastic net \cite{Tibshirani96,MeinshausenYu09,ZouHuiHastie:2005,Wipf_ARD_NIPS_2007,ABCP14}. However, in our framework the implementation of these approaches is not trivial and would require sophisticated distributed optimization algorithms like ADMM \cite{Boyd2011}. In fact, we consider a scenario where $N$ agents first collect a total of $M$ direct and noisy measurements of the unknown map on input locations drawn from a common and known probability density. The aim is then to obtain a shared function estimate. To simplify the exposition, we assume w.l.o.g.\ $N=M$, i.e., each agent collects a single measurement. We also assume that computational and data storage capabilities are limited, and that the communication network is peer-to-peer, i.e., agents are able only to communicate with a restricted number of neighbors. As described below, this makes the problem difficult also under Gaussian process assumptions, but we will see that function estimation 
can be performed using simple average operations. 

Assuming that $f$ and the measurements noise are jointly Gaussian, achieving the minimum variance estimate requires knowing all the $M$ measurements and related input locations, plus invert an $M \times M$ matrix with $O(M^3)$ operations, a difficult task in a distributed fashion. When the data set size $M$ is large, the complexity is high also in centralized contexts. Therefore, many alternative approaches have been developed relying, e.g., on the notion of pseudo input locations~\cite{Qui2005,Snelson06sG,Laz2010}, the use of matrix factorizations~\cite{Ambi2016} and approximations of the kernel function~\cite{BachLowRank05,KulisLowRank06} through the Nystr\"{o}m method or greedy techniques~\cite{Williams2001,ZhangNistrom10,SmolaGeedy2000}. Along this way, \ac{KL} expansions~\cite{levy2008karhunen} have been also used to decompose the kernel in terms of eigenfunctions that are orthogonal w.r.t.\ the input locations probability density. One can then approximate the Gaussian process via the $E$ kernel eigenfunctions associated to the largest eigenvalues, an approximation that corresponds to perform the best process approximation before seeing the data~\cite{levy2008karhunen} (see \Section~\ref{sec:reformulating_the_measurements_model_using_kl_expansions} for more details). A posteriori, i.e., after seeing the measurements and their input locations, the situation is instead more subtle since there exist $E$-dimensional subspaces that allow to come closer to the minimum variance estimator~\cite{ferrari_trecate_et_al__1999__finite-dimensional_approximation_of_gaussian_processes}. However, the a priori basis given by the \ac{KL} expansion has important advantages. In fact, as proved in~\cite{zhu_et_al__1998__gaussian_regression_and_optimal_finite_dimensional_linear_models}, the first $E$ kernel eigenfunctions are asymptotically optimal, i.e., they provide the best $E$-dimensional approximation of the minimum variance estimator as the data set size $M$ grows to infinity. In addition, differently from the a posteriori basis described in~\cite{ferrari_trecate_et_al__1999__finite-dimensional_approximation_of_gaussian_processes}, the a priori basis can be computed off-line.
Moreover, as detailed in \Section~\ref{sec:finite-dimensional_approximations_of_the_bayesian_estimator}, computing the final estimates requires computing sufficient statistics that have the structure of averages of $M$ local matrices and local vectors of dimension respectively $E \times E$ and $E$. This implies that the basic building block of the estimators involves computing averages over networks which can be more efficient from a memory, computation and communication perspective when $E \ll M$. Such averages can be computed via the so called \emph{average consensus algorithms}~\cite{Xiao:04,garin2010survey} which require only mild assumption on network connectivity and communication. In particular, these algorithms require no global topological information, only minimal local coordination and can be implemented also in the context of asynchronous updates and lossy communication~\cite{Bof:17}. 

\vspace{-2mm}
\subsection{Contribution}

Our stream of research pairs the ones of other authors focusing on distributed kernel regression. An example is~\cite{predd2009collaborative}, that proposes a distributed regularized kernel \ac{LS} regression algorithm that exploits successive orthogonal projections, or~\cite{perez2010robust} that extends \cite{predd2009collaborative} by designing strategies to reduce the communication and synchronization needs. Estimators with reduced order model complexity have been proposed in~\cite{honeine2008distributed}, while nonparametric schemes using Nearest-Neighbors interpolation strategies have been studied also in \cite{martinez__2010__distributed_interpolation_schemes_for_field_estimation_by_mobile_sensor_networks}. Another Gaussian estimation approach is considered in \cite{xu2013efficient}, with focus on the problem of sequentially predicting the most informative future input locations to minimize simultaneously the prediction error and the uncertainty in the regularization parameters. Other distributed regression algorithms are proposed in~\cite{cortes__2009__distributed_kriged_kalman_filter_for_spatial_estimation} with the aim of estimating a dynamic Gaussian process and its gradient, while in \cite{choi_et_al__2009__distributed_learning_and_cooperative_control_for_multi-agent_systems} authors develop a distributed learning and cooperative control algorithm where agents estimate a static field modeled as a network of radial basis functions whose centers locations are known in advance.

Despite the many research efforts, none of the aforementioned works on distributed regression have addressed the following fundamental issue: \emph{assigned a Gaussian prior (the kernel) and the input locations distribution, how much information does the network need to exchange to obtain, with a probability $1-\alpha$, the desired level of estimation accuracy?} In this paper we will answer this question adopting \ac{KL}-based strategies which exploit $E$ kernel eigenfunctions. In particular, we will study two different estimators denoted by $\widehat{f}_{A}$ and $\widehat{f}_{B}$ which have computational and communication complexities of order $O(E^2)$ and $O(E)$, respectively, originally proposed in~\cite{varagnolo2012distributed}. Differently from~\cite{varagnolo2012distributed} which focused on finding Monte Carlo based strategies for assessing the a posteriori statistical performance of the estimators, in this work the focus is on characterizing their a priori prediction capability on future data by first assigning the kernel and the input locations statistics, and then deriving non-asymptotic error bounds that are functions of $E$, $M$ and $\alpha$. This analysis can be also seen as the extension to the Bayesian context of the concept of \emph{effective dimension} developed in deterministic frameworks, e.g., in~\cite{Zhang2005}. There it has been shown that, in the worst case, subspaces of dimension $\sqrt M$, i.e., sub-polynomial in the data set size, capture the estimate. Parallel to this, our bound returns information on the \emph{Bayesian effective dimension} revealing which subspace can be really influenced by the measurements.

Another major contribution provided in this work is to show that both $\widehat{f}_{A}$ and $\widehat{f}_{B}$ are asymptotically optimal, i.e., for fixed $E$, as $M$ grows to infinity there is no other estimator which can perform better in the mean squared error sense. We will also see that, while $\widehat{f}_{A}$ is always consistent, i.e., convergent in probability to the true function as $E,M \rightarrow \infty$, consistency of $\widehat{f}_{B}$ requires $E$ to grow slower than $M$. In some sense, such result clarifies the price to pay when adopting a estimator parsimonious in the information exchange.

Finally, in many applications the kernel scale factor is unknown and its tuning is critical since it strongly affects the performance of the Bayesian estimator. In addition, the kernel expansion could be hard to be obtained and one would rather use a different set of basis functions. In the terminal part of the paper, we address these problems by proposing a novel distributed tuning strategy based on the \ac{SURE} criterion \cite{stein__1981__estimation_of_the_mean_of_a_multivariate_normal_distribution}. Standard approaches proposed in the literature in the context of a centralized framework (like cross-validation and maximum likelihood \cite{Rice1986,Golub79,Maritz:1989,Hastie09}, MAP estimation \cite{Xu:11}, expected improvement \cite{Snoek:12} and Markov chain Monte Carlo \cite{Gilks,Magni1998}) require high computation and communication overhead, and are therefore not suited for distributed implementations. Instead, our strategy allows for simultaneous hyperparameter tuning and function estimation via a single average consensus algorithm over a vector of size $O(E^2)$ when $\widehat{f}_{A}$ is employed, and via only two averages of size $O(E)$ when $\widehat{f}_{B}$ is employed. Very importantly, the \ac{SURE} criterion be used also for generic basis functions, such as kernel sections or Nystr{\"o}m bases, thus not necessarily restricted to be kernel eigenfunctions and defined.

\vspace{-2mm}
\subsection{Paper outline}
 
The paper is organized as follows. \Section~\ref{sec:the_bayesian_estimation_problem} formulates the Bayesian estimation problem while \Section~\ref{sec:KLEstimators} describes the \ac{KL} expansion of the Gaussian process and the distributed estimators. \Section~\ref{sec:statistical_characterizations_of_widehat_f_a_and_widehat_f_b} provides the statistical characterization of our distributed estimators, also deriving error bounds which are then tested via some numerical experiments. 
\Section~\ref{sec:distributed_tuning_of_the_hyperparameters_of_widehat_f_a} proposes distributed strategies to tune the possibly unknown regularization parameter entering our estimators for generic basis functions and discusses practical implementation issues. These strategies are also tested on both synthetic and real data. \Section~\ref{Conclusions} collects conclusions and future research directions while proofs are collected in the Appendix.

\section{Bayesian estimation}
\label{sec:the_bayesian_estimation_problem}

\subsection{The measurements model}
\label{ssec:measurements_model}

We consider the measurements model 
\begin{equation}
	y_{m} = f \left( x_{m} \right) + \nu_{m},
	\qquad
	m = 1, \ldots, M
	\label{equ:measurement_model}
\end{equation}
with the input locations $x_{m}$ following the stochastic generation scheme
\begin{equation}
	x_{m} \sim \mu(\mathcal{X}) \; \text{ i.i.d.},
	\qquad
	m = 1, \ldots, M,
	\label{equ:input_locations_model}
\end{equation}
with $\mu$ a non-degenerate probability measure on the compact $\mathcal{X}$. The unknown function $f: \mathcal{X} \rightarrow \Reals$ is a zero-mean Gaussian process with continuous covariance $K: \mathcal{X} \times \mathcal{X} \rightarrow \Reals$, i.e.,
\begin{equation}
	f \sim \GaussianDistribution{0}{K}.
	\label{equ:distribution_of_the_estimand}
\end{equation}
The measurement noise is also Gaussian of known variance $\sigma^{2}_{\nu}$:
\begin{equation}
	\nu_{m} \sim \GaussianDistribution{0}{\sigma^{2}_{\nu}}.
\end{equation}
Finally, $\{ \nu_{m} \}_{m=1}^M$, $\{ x_{m} \}_{m=1}^M$ and $f$ are all assumed mutually independent.

\subsection{The Bayesian estimator}
\label{ssec:the_bayesian_estimator}

The Gaussian assumptions of \Section~\ref{ssec:measurements_model} imply that the posterior of $f$ given the dataset $\left\{ x_{m}, y_{m} \right\}_{m = 1}^{M}$ is still Gaussian. Also, the \ac{MAP} estimator coincides with the minimum variance estimator and is given by
\begin{equation}
	\widehat{f}_{\MAP}(x)
	=
	\begin{bmatrix}
		K(x, x_{1}) \;
		\ldots \;
		K(x, x_{M})
	\end{bmatrix}
	H_{\MAP}
	\begin{bmatrix}
		y_{1} \\
		\vdots \\
		y_{M} \\
	\end{bmatrix}
	\label{equ:MAP_estimator}
\end{equation}
with
\begin{equation}
	H_{\MAP}
	\DefinedAs
	\left(
		\begin{bmatrix}
			K(x_{1}, x_{1}) & \cdots & K(x_{1}, x_{M}) \\
			\vdots			&		 & \vdots \\
			K(x_{M}, x_{1}) & \cdots & K(x_{M}, x_{M}) \\
		\end{bmatrix}
		+
		\sigma^{2}_{\nu}
		I
	\right)^{-1} .
	\label{equ:definition_of_H_MAP}
\end{equation}

The storage and computational requirements needed to compute $\widehat{f}_{\MAP}$ are thus $\BigOOf{M^{2}}$ and $\BigOOf{M^{3}}$, respectively. The communication complexity is either $\BigOOf{\dim(\mathcal{X}) M}$ if agents share the input locations $x_{m}$ or $\BigOOf{M^{2}}$ if they share the covariances $K(x_{m}, x_{m'})$. Thus, storage, computational and communication complexities do not scale favorably with the dataset size $M$. Our aim is thus to find good approximators of $\widehat{f}_{\MAP}$ that are suitable for distributed implementations.

\section{Finite-dimensional approximations of the Bayesian estimator}
\label{sec:KLEstimators}

\subsection{\ac{KL} expansion: kernel}
\label{sec:reformulating_the_measurements_model_using_kl_expansions}

The kernel \eqref{equ:distribution_of_the_estimand} can be expanded in terms of eigenfunctions $\phi_{e}$ orthonormal w.r.t.\ the measure $\mu$ in~\eqref{equ:input_locations_model} and related eigenfunctions $\lambda_{e}$~\cite{levy2008karhunen}. They are defined by
\begin{equation}
	\lambda_{e} \phi_{e} (x) = \int_{\mathcal{X}} K(x,x') \phi_{e} (x') d\mu(x'), 
	\label{equ:definition_of_eigenfunction}
\end{equation}
\begin{equation}
	K(x,x') = \sum_{e = 1}^{+\infty} \lambda_{e} \phi_{e}(x) \phi_{e}(x')
	\qquad
	\lambda_1 \geq \lambda_2 \ldots > 0,
	\label{equ:eigenexpansion_of_K}
\end{equation}
and, using $\delta_{ij}$ for the Kronecker delta,
\begin{equation}
	\int_{\mathcal{X}} \phi_{i}(x) \phi_{j}(x) d \mu(x) = \delta_{ij}. 
	\label{equ:orthogonality_of_the_eigenfunctions}
\end{equation}

Let $E$ be a positive integer. Then~\eqref{equ:definition_of_eigenfunction}, \eqref{equ:eigenexpansion_of_K} and~\eqref{equ:orthogonality_of_the_eigenfunctions} allow us to reformulate the process $f$ via the following \ac{KL} expansion
\begin{equation}
	f(x)
	=
	\underbrace{\sum_{e = 1}^{E} a_{e} \phi_{e}(x)}_{\IDefinedAs \displaystyle f_{a}(x)}
	\; + \;
	\underbrace{\sum_{e = 1}^{+\infty} b_{e} \phi_{E+e}(x)}_{\IDefinedAs \displaystyle f_{b}(x)} .
	\label{equ:estimand_model}
\end{equation}
The expansion coefficients have been thus divided into two sets: a finite one composed by the $E$ random variables $a_e$, and an infinite one given by the remaining variables $b_e$. The elements in these two sets are all mutually independent, and satisfy
\begin{subequations}{\label{equ:coefficients_are_gaussianly_distributed}}
	\noeqref{equ:coefficients_are_gaussianly_distributed,equ:coefficients_are_gaussianly_distributed:a,equ:coefficients_are_gaussianly_distributed:b}
	\begin{gather}
		a_{e} \sim \GaussianDistribution{0}{\lambda_{e}}, \; e = 1, \ldots, E 
		\label{equ:coefficients_are_gaussianly_distributed:a} \\
		b_{e} \sim \GaussianDistribution{0}{\lambda_{E+e}}, \; e = 1, 2, \ldots
		\label{equ:coefficients_are_gaussianly_distributed:b} \\
	\end{gather}
\end{subequations}
It is well known that
\begin{equation}
	\mathcal{S}
	\DefinedAs
	\SpanOf{ \phi_{1}(\cdot), \ldots, \phi_{E}(\cdot) }
	\label{equ:definition_of_mathcal_S}
\end{equation}
is that $E$-dimensional subspace that captures the biggest part of the statistical energy of $f$ as measured by $\ExpectationOf{\int f^2 d\mu}$. In other words, $f_{a}$ is the best $E$-dimensional approximation of $f$ in the mean square sense~\cite{zhu_et_al__1998__gaussian_regression_and_optimal_finite_dimensional_linear_models}.

In what follows, it is always assumed that all the kernel eigenfunctions are contained in a ball of finite radius in the space of continuous functions, i.e.,
\begin{assumption}
	There exists a $k < +\infty$ s.t.
	\begin{equation}
		\sup_{x \in \mathcal{X}}
		\left| \phi_{e}(x) \right|
		\leq
		\sqrt{k}
		<
		+\infty 
		\qquad
		e = 1, 2, \ldots .
		\label{equ:eigenfunctions_are_uniformly_bounded}
	\end{equation}
	\label{ass:eigenfunctions_are_uniformly_bounded}
\end{assumption}

\Assumption~\ref{ass:eigenfunctions_are_uniformly_bounded} is satisfied by all the finite-dimensional kernels and also by classical covariances like the spline kernels, e.g., see \cite{Bell2004} for the case of uniform $\mu$. In practice, if the \ac{KL} expansion is not available in closed form, it can be obtained numerically with arbitrary accuracy, as for example described in~\cite{de1999consistent}, also permitting to compute the constant $k$.
 
\subsection{KL expansion: measurement model}
\label{sec:finite-dimensional_approximations_of_the_bayesian_estimator}

Our next step is to search for finite-dimensional estimators of $f$ suitable for distributed implementations. Below, we introduce two different estimators, denoted by $\widehat{f}_{A}$ and $\widehat{f}_{B}$, which assume values in the finite-dimensional subspace $\mathcal{S}$ defined in~\eqref{equ:definition_of_mathcal_S}. First, it is useful to rewrite model~\eqref{equ:measurement_model} in a more compact form.

Let
\begin{equation}
	\bm{x} \DefinedAs \left[ x_{1}, \ldots, x_{M} \right]^{T}
	\label{equ:definition_of_bm_x}
\end{equation}
\begin{equation}
	\bm{y} \DefinedAs \left[ y_{1}, \ldots, y_{M} \right]^{T}
	\qquad
	\bm{\nu} \DefinedAs \left[ \nu_{1}, \ldots, \nu_{M} \right]^{T}
	\label{equ:definition_of_bm_y_and_bm_nu}
\end{equation}
\begin{equation}
	\bm{a} \DefinedAs \left[ a_{1}, \ldots, a_{E} \right]^{T}
	\qquad
	\bm{b} \DefinedAs \left[ b_{1}, b_{2}, \ldots \right]^{T}
	\label{equ:definition_of_bm_a_and_bm_b}
\end{equation}
\begin{equation}
	G
	\DefinedAs
	\begin{bmatrix}
		G_{11} & \ldots & G_{1E} \\
		\vdots & & \vdots \\
		G_{M1} & \ldots & G_{ME} \\
	\end{bmatrix}
	\qquad
	Z
	\DefinedAs
	\begin{bmatrix}
		Z_{11} & Z_{12} & \ldots \\
		\vdots & & \vdots \\
		Z_{M1} & Z_{M2} & \ldots \\
	\end{bmatrix}
	\label{equ:definition_of_G_and_Z}
\end{equation}
\begin{equation}
	G_{me}
	\DefinedAs
	\phi_{e}(x_{m}),
	\quad
	m = 1, \ldots, M, \;
	e = 1, \ldots, E, \;
	\label{equ:definition_of_G}
\end{equation}
\begin{equation}
	Z_{me}
	\DefinedAs
	\phi_{E+e}(x_{m}),
	\quad
	m = 1, \ldots, M, \;
	e = 1, 2, \ldots 
	\label{equ:definition_of_Z}
\end{equation}
Considering decomposition~\eqref{equ:estimand_model}, definitions~\eqref{equ:definition_of_bm_y_and_bm_nu}-\eqref{equ:definition_of_Z} and using classical algebraic notation to handle infinite-dimensional objects, the measurements model~\eqref{equ:measurement_model} becomes

\begin{equation}
	\bm{y} = G \bm{a} + Z \bm{b} + \bm{\nu}.
	\label{equ:measurement_model_vector_form}
\end{equation}
With this novel notation $G \bm{a}$ accounts for the contribution from $f_{a}$ while $Z \bm{b}$ accounts for the contribution from $f_{b}$.

\subsection{The $E$-dimensional estimator $\widehat{f}_{A}$}
\label{ssec:the_finite-dimensional_estimator_widehat_f_a}

Let
\begin{equation}
	\widehat{f}_{A}(x)
	\DefinedAs
	\begin{bmatrix}
		\phi_{1}(x) \; \cdots \; \phi_{E}(x)
	\end{bmatrix}
	H_{A}
	\bm{y}
	\label{equ:estimator_A}
\end{equation}
where
\begin{equation}
	H_{A}
	\DefinedAs
	\left( \frac{G^{T} G}{M} + \frac{\sigma^{2}_{\nu}}{M} \Lambda^{-1}_{E} \right)^{-1} \frac{G^{T}}{M} 
	\label{equ:definition_of_H_A}
\end{equation}
and $\Lambda_{E} \DefinedAs \DiagonalMatrixOf{\lambda_{1}, \ldots, \lambda_{E}}$. The estimator $\widehat{f}_{A}$ is suitable for distributed computations. In fact, defining
\begin{equation}
	G_{m}
	\DefinedAs
	\left[ \phi_{1}(x_{m}), \ldots, \phi_{E}(x_{m}) \right]
	\label{equ:definition_of_G_m}
\end{equation}
one has
\begin{equation}
	\frac{G^{T} G}{M}
	=
	\frac{1}{M}\sum_{m = 1}^{M} G_{m}^{T} G_{m},
	\ \
	\frac{G^{T} \bm{y}}{M}
	=
	\frac{1}{M} \sum_{m = 1}^{M} G_{m}^{T} y_{m} .
	\label{equ:quantities_in_estimator_A_are_average_consensi}
\end{equation}
Since $G_{m}^{T} G_{m} \in \Reals^{E \times E}$ and $G_{m}^{T} y_{m} \in \Reals^{E}$ are local quantities, \eqref{equ:quantities_in_estimator_A_are_average_consensi} points out that $\widehat{f}_{A}$ can be distributedly computed through the parallelization of two average consensus strategies: one on the $G_{m}^{T} G_{m}$'s and one on the $G_{m}^{T} y_{m}$'s, for a total of $E^{2} + E$ scalars. This estimator would correspond to the \ac{MVUE} estimator if the process $f$ in~\eqref{equ:estimand_model} were truncated just to $f_{a}$.

\subsection{The $E$-dimensional estimator $\widehat{f}_{B}$}
\label{ssec:the_e-dimensional_estimator_widehat_f_b}

As stated in~\eqref{equ:orthogonality_of_the_eigenfunctions}, one has
\begin{equation}
	\ExpectationOf{ \left[ \frac{G^{T} G}{M} \right]_{e, e'} } 
	=
	\int_{\mathcal{X}} \phi_{e}(x) \, \phi_{e'}(x) \, d \mu(x) 
	=
	\delta_{e, e'}.
\end{equation}
and, given the assumptions in \Section~\ref{ssec:measurements_model} and \Assumption~\ref{ass:eigenfunctions_are_uniformly_bounded}, 
the following convergence in probability holds

\begin{equation}
	\frac{G^{T} G}{M}
	=
	\frac{1}{M} \sum_{m = 1}^{M} G_{m}^{T} G_{m}
	\xrightarrow{M \rightarrow +\infty}
	\ExpectationOf{ \frac{G^{T} G}{M} }
	=
	I.
	\label{equ:expected_G_T_G_over_M_equals_I}
\end{equation}

Thus, it is tempting to use the approximation
\begin{equation}
	\frac{G^{T} G}{M} 
	\approx 
	I
	\label{equ:expected_G_T_G_over_M_equals_I_2}
\end{equation}
and use, in place of $H_{A}$ in~\eqref{equ:estimator_A}, the matrix 
\begin{equation}
	H_{B}
	\DefinedAs
	\left( I + \frac{\sigma^{2}_{\nu}}{M} \Lambda^{-1}_{E} \right)^{-1} \frac{G^{T}}{M}.
	\label{equ:definition_of_H_B}
\end{equation}

In turn, this approach approximates $\widehat{f}_{A}$ with
\begin{equation}	
	\widehat{f}_{B}(x)
	\DefinedAs
	\begin{bmatrix}
		\phi_{1}(x) \; \cdots \; \phi_{E}(x)
	\end{bmatrix}
	H_{B}
	\bm{y} .	
	\label{equ:estimator_B}
\end{equation}
The estimator $\widehat{f}_{B}$ is more advantageous than $\widehat{f}_{A}$ for distributed computations. In fact, it requires an average consensus on just the column vectors $G_{m}^{T} y_{m}$'s, for a total of $E$ scalars (differently from the $E^{2} + E$ ones required by $\widehat{f}_{A}$), and does not require any expensive matrix inversion since $I + \frac{\sigma^{2}_{\nu}}{M} \Lambda^{-1}_{E}$ is diagonal.

\section{Statistical analysis of $\widehat{f}_{A}$ and $\widehat{f}_{B}$}
\label{sec:statistical_characterizations_of_widehat_f_a_and_widehat_f_b}

Ideally one would like to compute $\ExpectationOf{\| f - \widehat{f}_{A} \|^2}$ and $\ExpectationOf{\| f - \widehat{f}_{B} \|^2}$, or at least some bounds that quantify the performance of the estimator for any specified $E$ and $M$ a priori. However, the computation of such quantities is intractable or, at least, requires an expensive Monte Carlo analysis, possibly to be repeated for many different design variables like, e.g., $M,E,\sigma^{2}_{\nu}$. To circumvent this challenge, we will exploit the assumption that the input locations are randomly drawn from a known distribution $\mu$ and the orthonormality of the eigenfunctions to find bounds on $\ExpectationOf{\| f - \widehat{f}_{A} \|^2}$ that hold with arbitrarily high probability. More specifically, the key idea is to find an event $\mathcal{E}$ that occurs with arbitrarily high probability such that informative bounds on $\ExpectationOfGiven{\| f - \widehat{f}_{A} \|^2}{\mathcal{E}}$ can be computed. This is formally described in the next sections.

\subsection{Performance indexes and lower bound}

Two important performance indexes we consider for $\widehat{f}_{A}$ and $\widehat{f}_{B}$ are the errors defined by the conditional expectations
\begin{equation}
	\begin{array}{l}
		\displaystyle
		\mathrm{Err}_{A}(\bm{x}) \DefinedAs \ExpectationOfGiven{\left\| f - \widehat{f}_{A} \right\|^{2}}{\bm{x}} \vspace{0.1cm} \\
		\displaystyle
		\mathrm{Err}_{B}(\bm{x}) \DefinedAs \ExpectationOfGiven{\left\| f - \widehat{f}_{B} \right\|^{2}}{\bm{x}}
	\end{array}
	\label{equ:definition_of_Err_A_and_Err_B}
\end{equation}
where
\begin{equation}
	\| g \|^{2} \DefinedAs \int_{\mathcal{X}} g^2(x) d \mu(x) .
	\label{equ:definition_of_norm_induced_by_mu}
\end{equation}
The variables $\mathrm{Err}_{A}(\bm{x})$ and $\mathrm{Err}_{B}(\bm{x})$ are stochastic, since they are functions of the random input locations $\bm{x}$ that in our settings are assumed random as described in~\eqref{equ:input_locations_model}. Hence, the crux of our analysis will be how to account for the randomness coming from $\bm{x}$. Note also that $\| \cdot \|$ depends on $\mu$ so that $\mathrm{Err}_{A}$ and $\mathrm{Err}_{B}$ quantify the prediction errors on future data independently drawn from the same training set distribution.

Exploiting the \ac{KL} expansion introduced in \Section~\ref{sec:reformulating_the_measurements_model_using_kl_expansions} a lower bound on the errors $\mathrm{Err}_{A}(\bm{x})$ and $\mathrm{Err}_{B}(\bm{x})$ can be also easily obtained.
More generally, the following result bounds the performance achievable by any generic 
$E$-dimensional estimator of $f$.
%
\begin{theorem}
	Let $\widehat{f}_{\star}$ be any generic estimator of $f$, function of $\bm{x}$ and $\bm{y}$ and assuming values in any generic $E$-dimensional space fixed a priori. Then 
	\begin{equation}
		\min_{\widehat{f}_{\star}} \;
		\ExpectationOfGiven{\| f - \widehat{f}_{\star} \|^{2}}{\bm{x}} 
		\geq
		\sum_{e = E + 1}^{+\infty} \lambda_{e}.
	\label{equ:Err_lower_bound}
	\end{equation}
	\label{thm:lower_bound_for_generic_estimator_of_f}
\end{theorem}

The following definition will be especially important for our future developments.

\begin{definition}
	\label{def:BarE}
	We say that $\overline{\mathrm{Err}}_{A} \leq q$ or $\overline{\mathrm{Err}}_{B} \leq q$ with probability $1 - \alpha$ if there exists an event $\mathcal{E}$ in the $\sigma$-algebra induced by $\bm{x}$ of probability at least $1-\alpha$ such that, respectively,
	\begin{equation}
		\ExpectationOf{ \mathrm{Err}_{A}(\bm{x}) \left| \; \mathcal{E} \right. } \leq q
		\label{equ:definition-of-overline-Err-A}
	\end{equation}
	or 
	\begin{equation}
		\ExpectationOf{ \mathrm{Err}_{B}(\bm{x}) \left| \; \mathcal{E} \right. } \leq q.
		\label{equ:definition-of-overline-Err-B}
	\end{equation}
\end{definition}

Thus, if $\alpha$ is close to zero saying that $\overline{\mathrm{Err}}_{A} \leq q$ with probability $1-\alpha$ is equivalent to saying that the average error associated to $\widehat{f}_{A}$ is smaller than $q$ with high probability. Finally, note that setting $\mathcal{E}$ to the entire sample space, the conditional expectations in the \ac{LHS} of \eqref{equ:definition-of-overline-Err-A} and \eqref{equ:definition-of-overline-Err-B} become unconditional ones, and actually correspond to the \acp{MSE} of $\widehat{f}_{A}$ and $\widehat{f}_{B}$, i.e.,
\begin{equation} \label{MSEA}
\text{MSE}_{\widehat{f}_{A}} = \int_{\mathcal{X}} \ \mathrm{Err}_{A}(x) d\mu(x), 
\end{equation}
\begin{equation} \label{MSEB}
\text{MSE}_{\widehat{f}_{B}} = \int_{\mathcal{X}} \ \mathrm{Err}_{B}(x) d\mu(x).
\end{equation}

\subsection{Non asymptotic error bounds}

\begin{figure*}
	\centering
	\pgfplotsset
{
	AxisStyle/.style =
	{
		width					= 0.99\columnwidth,
		height					= 0.5\columnwidth,
		xmin					= 1,
		xmax					= 100,
		ymin					= 0,
		ymax					= 10,
		xlabel					= {$E$},
		xlabel near ticks,
		xlabel shift			= 3,
		ylabel					= {},
		ylabel near ticks,
		scaled y ticks			= false,
		x tick label style		= {anchor = north, inner ysep = 0.1cm},
		y tick label style		= {/pgf/number format/fixed, inner xsep = 0.1cm},
	},
	LegendStyle/.style =
	{
		legend plot pos			= left,
		legend columns			= 1,
		legend style			= {draw = black, fill = white},
		legend cell align		= left,
		inner xsep				= 0.0cm,
		inner ysep				= 0.0cm,
		legend style			=
		{
			nodes				= {font = \scriptsize},
			at					= {(0,0)},
			anchor				= south west,
		}
	},
	GridStyle/.style =
	{
		xmajorgrids,
		ymajorgrids,
		xminorgrids,
		minor grid style =
		{
			thin,
			densely dotted,
			black!20
		},
		major grid style =
		{
			thin,
			densely dotted,
			black!20
		},
	},
	BndStyle/.style =
	{
		smooth,
		mark			= none,
		solid,
		draw			= black!60!white,
		line width		= 0.05cm,
	},
	TrueErrorStyle/.style =
	{
		BndStyle,
		dashed,
		draw			= black!80!white,
		line width		= 0.03cm,
	},
	LowerBndStyle/.style =
	{
		BndStyle,
		solid,
		draw			= black,
		line width		= 0.01cm,
	},
}

\begin{tikzpicture}
	\begin{axis}
	[
		AxisStyle,
		LegendStyle,
		GridStyle,
		ymax		= 1,
		title		= {$K = $ splines, $\mathrm{Bnd}_{A}$},
		xmode		= log, 
		ymode		= log, 
		name		= splinesBndA,
	]
		\addplot [BndStyle] table [x = E, y = NormalizedBound]
		{data-BoundsKSplinesA.txt};
		\addlegendentry{normalized $\mathrm{Bnd}_{A}(E)$};
		\addplot [TrueErrorStyle] table [x = E, y = TrueError]
		{data-BoundsKSplinesA.txt};
		\addlegendentry{normalized error};
		\addplot [LowerBndStyle] table [x = E, y = LowerBound]
		{data-BoundsKSplinesA.txt};
		\addlegendentry{lower bound};
	\end{axis}
	\begin{axis}
	[
		AxisStyle,
		GridStyle,
		title		= {$K = $ splines, $\mathrm{Bnd}_{B}$},
		xmode		= log, 
		ymode		= log, 
		ymax		= 1,
		name		= splinesBndB,
		at			= {($(splinesBndA.east)+(2cm,0)$)},
		anchor		= west,
		LegendStyle,
	]
		\addplot [BndStyle] table [x = E, y = NormalizedBound]
		{data-BoundsKSplinesB.txt};
		\addlegendentry{normalized $\mathrm{Bnd}_{B}(E)$};
		\addplot [TrueErrorStyle] table [x = E, y = TrueError]
		{data-BoundsKSplinesB.txt};
		\addlegendentry{normalized error};
		\addplot [LowerBndStyle] table [x = E, y = LowerBound]
		{data-BoundsKSplinesB.txt};
		\addlegendentry{lower bound};
	\end{axis}
	\begin{axis}
	[
		AxisStyle,
		GridStyle,
		title		= {$K = $ exponential decay, $\mathrm{Bnd}_{A}$},
		xmode		= log, 
		ymode		= log, 
		name		= exponentialBndA,
		at			= {($(splinesBndA.south)-(0,2cm)$)},
		anchor		= north,
		LegendStyle,
	]
		\addplot [BndStyle] table [x = E, y = NormalizedBound]
		{data-BoundsKExponentialA.txt};
		\addlegendentry{normalized $\mathrm{Bnd}_{A}(E)$};
		\addplot [TrueErrorStyle] table [x = E, y = TrueError]
		{data-BoundsKExponentialA.txt};
		\addlegendentry{normalized error};
		\addplot [LowerBndStyle] table [x = E, y = LowerBound]
		{data-BoundsKExponentialA.txt};
		\addlegendentry{lower bound};
	\end{axis}
	\begin{axis}
	[
		AxisStyle,
		GridStyle,
		title		= {$K = $ exponential decay, $\mathrm{Bnd}_{B}$},
		xmode		= log, 
		ymode		= log, 
		name		= expM100000,
		at			= {($(splinesBndB.south)-(0,2cm)$)},
		anchor		= north,
		LegendStyle,
	]
		\addplot [BndStyle] table [x = E, y = NormalizedBound]
		{data-BoundsKExponentialB.txt};
		\addlegendentry{normalized $\mathrm{Bnd}_{B}(E)$};
		\addplot [TrueErrorStyle] table [x = E, y = TrueError]
		{data-BoundsKExponentialB.txt};
		\addlegendentry{normalized error};
		\addplot [LowerBndStyle] table [x = E, y = LowerBound]
		{data-BoundsKExponentialB.txt};
		\addlegendentry{lower bound};
	\end{axis}
\end{tikzpicture}
	\caption{$\mathrm{Bnd}_{A}$ and $\mathrm{Bnd}_{B}$ (normalized by the a priori function variance) as a function of $E$, with $\alpha = 0.05,M=10000$ and for different eigenvalues decay rates.}
\label{fig:bounds}
\end{figure*}
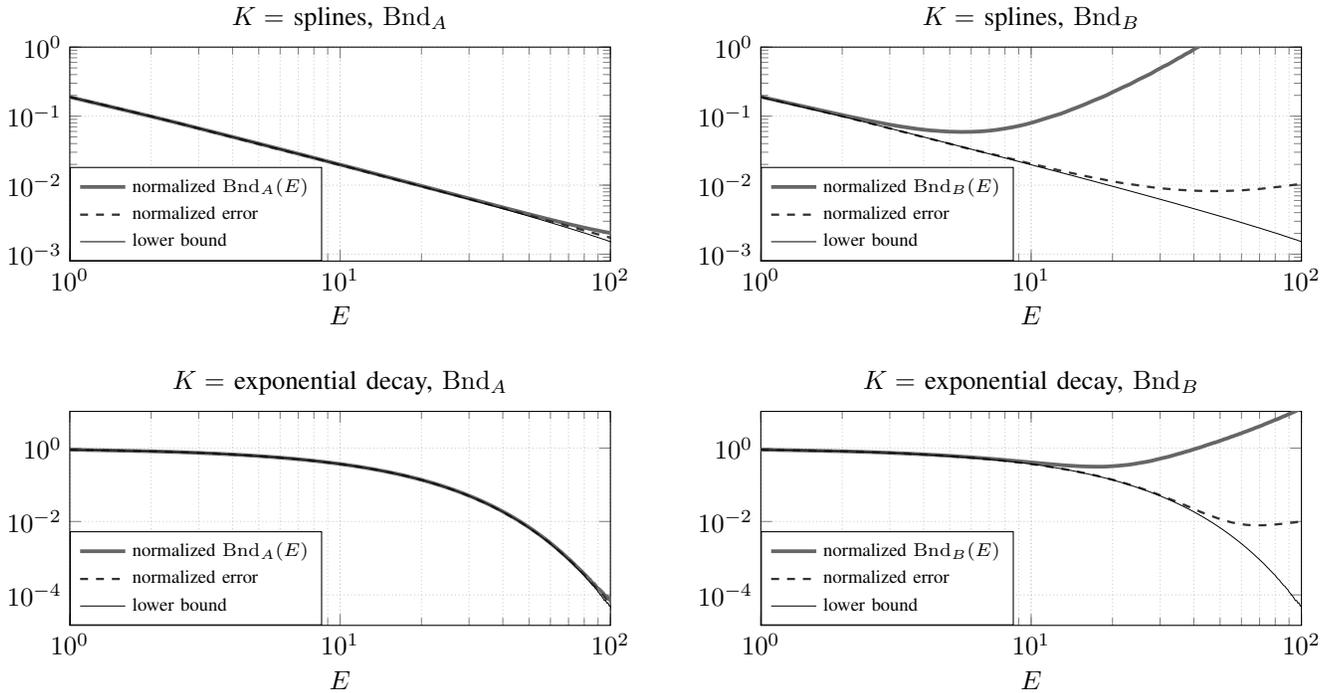

The key issue is to bound the performance indexes $\mathrm{Err}_{A}$ and $\mathrm{Err}_{B}$ for any finite number of measurements $M$ and eigenfunctions $E$. The following theorem provides the desired bounds. It depends on the input locations
distribution $\mu$, the kernel eigenvalues $\lambda_{e}$ and constant $k$ defined 
in~\eqref{equ:definition_of_eigenfunction} and \eqref{equ:eigenfunctions_are_uniformly_bounded},
 the number of eigenfunctions $E$ and measurements $M$. In addition the bound is
also function of a parameter $\varepsilon \in (0, 1]$ connected to maximal and minimal (stochastic) eigenvalue of $\frac{G^{T} G}{M}$, 
as detailed in the proof contained in the Appendix. 
\begin{theorem}
	Let the assumptions in \Section~\ref{ssec:measurements_model} and \Assumption~\ref{ass:eigenfunctions_are_uniformly_bounded} hold, $\alpha \in (0, 1)$ be a desired confidence level (e.g., $0.01$ or $0.05$), and $\varepsilon \in (0, 1]$ be given. If $E, M$ and $k$ satisfy
	\begin{equation}
		1 - \varepsilon + \varepsilon \log(\varepsilon)
		\geq
		\frac{Ek}{M} \log\left(\frac{E}{\alpha}\right)
		\label{equ:condition_on_varepsilon_for_estimator_A}
	\end{equation}
	then with probability at least $1-\alpha$ it holds that
	\begin{equation}
		\overline{\mathrm{Err}}_{A} 
		\leq
		\mathrm{Bnd}_{A}
		\label{equ:definition_of_bound_A}
	\end{equation}
	with
	\begin{equation}
		\begin{array}{l}
			\displaystyle
			\mathrm{Bnd}_{A}
			\DefinedAs
			%
			\frac{k M}{1-\alpha}
			\left(
				\sum_{e = 1}^{E}
				\frac{\lambda^2_{e}}{(\varepsilon M \lambda_e + \sigma^{2}_{\nu})^2}
			\right)
			\left(
				\sum_{e = E + 1}^{+\infty} \lambda_{e}
			\right) \\
			%
			\displaystyle \qquad
			+
			\frac{\sigma^{2}_{\nu}}{1-\alpha}
			\left(
				\sum_{e = 1}^{E}
				\frac{\lambda_e}{\varepsilon M \lambda_e + \sigma^{2}_{\nu}}
			\right)
			%
			+
			\left(
				\sum_{e = E + 1}^{+\infty} \lambda_{e}
			\right)
			.
		\end{array}
		\label{equ:value_of_bound_A}
	\end{equation}

	Under the same assumption but with $E, M$ and $k$ now satisfying
	\begin{equation}
		1 - \varepsilon + \varepsilon \log(\varepsilon)
		\geq
		\frac{Ek}{M} \log\left(\frac{2 E}{\alpha}\right),
		\label{equ:condition_on_varepsilon_for_estimator_B}
	\end{equation}
	then with probability at least $1-\alpha$ it holds that
	\begin{equation}
		\overline{\mathrm{Err}}_{B}
		\leq
		\mathrm{Bnd}_{B}
		\label{equ:definition_of_bound_B}
	\end{equation}
	with
	\begin{equation}
		\begin{array}{l}
			\displaystyle
			\mathrm{Bnd}_{B}
			\DefinedAs
			%
			\frac{k M}{1 - \alpha}
			\left(
				\sum_{e = 1}^{E} \frac{\lambda_{e}^{2}}{\left( M \lambda_{e} + \sigma^{2}_{\nu} \right)^{2}}
			\right)
			\left( \sum_{e = E + 1}^{+\infty} \lambda_{e} \right)
			\\ \qquad \displaystyle
			%
			+
			\frac{\sigma^{2}_{\nu}}{1-\alpha}
			\left(
				\sum_{e = 1}^{E}
				\frac{\lambda_e}{\varepsilon M \lambda_e + \sigma^{2}_{\nu}}
			\right)
			%
			+
			\left(
				\sum_{e = E + 1}^{+\infty} \lambda_{e}
			\right)
			\\ \qquad \displaystyle
			+
			\kappa
			\left(
				\frac{E}{M}
				\sigma_{\nu}^{2}
				+
				\sum_{e = 1}^{E} \lambda_{e}
			\right)
			%
		\end{array}
		\label{equ:value_of_bound_B}
	\end{equation}
	where
	\begin{equation}
		\kappa
		=
		\frac{1}{1-\alpha}
		\left( \varepsilon + \frac{\lambda_1^{-1} \sigma_{\nu}^2}{M} \right)^{-4}
		(1-\varepsilon)^2 (2-\varepsilon)^2
	\label{equ:definition_of_kappa}
	\end{equation}
\label{thm:bounds_on_errors_of_f_A_and_f_B}
\end{theorem}

The obtained bounds are now tested via a numerical example.

\subsection{Numerical study}
\label{exa:increasing_E_can_lead_to_a_performance_loss}
	Consider the first-order spline kernel \cite{Wahba1990} which corresponds to the Brownian motion covariance, i.e.,
	\begin{equation}
		K(x,x') = \min(x,x') = \sum_{e=1}^\infty \lambda_{e} \phi_{e}(x) \phi_{e} (x')
	\end{equation}
	with the input locations probability measure $\mu$ in~\eqref{equ:input_locations_model} set to the uniform distribution on $[0,1]$. With these settings
	\begin{equation}
		\label{Exp1}
	 \phi_{e}(x) = \sqrt{2} \sin \left(x (e\pi - \pi/2)\right), \quad \lambda_e = \frac{1}{(e\pi - \pi/2)^2} 
	\end{equation}
	and $k=2$. To make the bounds only depend on $E$ we set $M=10000$, $1-\alpha=0.95$, the noise variance $\sigma^{2}_{\nu}=0.1^2$, and $\varepsilon \in (0, 1]$ that minimizes the bound while satisfying \eqref{equ:condition_on_varepsilon_for_estimator_A} or~\eqref{equ:condition_on_varepsilon_for_estimator_B} accordingly.

	The thick lines in the two top panels of \Figure~\ref{fig:bounds} show how $\mathrm{Bnd}_{A}$ (left) and $\mathrm{Bnd}_{B}$ (right) vary with $E$ (bounds are normalized using the prior process variance $\sum_{e=1}^\infty \lambda_{e}$). For the sake of comparison we also display the true (normalized) \acp{MSE} (dashed line) as defined in \eqref{MSEA} and \eqref{MSEB}, calculated via a Monte Carlo of 1000 runs, and its lower bound (thin line), i.e., $\sum_{e=E+1}^\infty \lambda_{e}/\sum_{e=1}^\infty \lambda_{e}$ as illustrated in \Theorem~\ref{thm:lower_bound_for_generic_estimator_of_f}.

As for $\mathrm{Bnd}_{A}$, it is interesting to notice that just 20 eigenfunctions are needed to obtain an high estimation accuracy in both the cases. In addition, the curve is very close to the true error profile (which in turn is close to the lower bound) and is monotonically decreasing. Indeed, as discussed in the proof of \Theorem~\ref{thm:f_A_and_f_B_are_consistent} contained in the next subsection, when one adopts $\widehat{f}_{A}$ one should set $E$ as large as possible (compatibly with communication capabilities) since, at the limit, convergence to the minimum variance estimator holds.

The profile of $\mathrm{Bnd}_{B}$ is instead different and exhibit a clear minimum at $E=7$. The reason is that $\widehat{f}_{B}$ relies on the asymptotic matrix approximation~\eqref{equ:expected_G_T_G_over_M_equals_I_2}. The bound $\mathrm{Bnd}_{B}$ then points out that if $E$ is too large then the quality of this approximation can worsen, hence leading to an increment of the corresponding \ac{MSE}. One can see that also the true error profile is not monotonically decreasing (indeed, we will see in the next subsection that for $M$ fixed and $E$ going to infinity $\widehat{f}_{B}$ is not guaranteed to converge to the minimum variance estimator). Note that, in this case, $\mathrm{Bnd}_{B}$ is close to truth only for low values of $E$ and that the Monte Carlo analysis suggests the best $E$ to be around 50. Overall, this indicates that the eigenfunctions number has to be seen as an important design parameter for $\widehat{f}_{B}$ to optimize the performance. This point will be the focus of \Section~\ref{sec:distributed_tuning_of_the_hyperparameters_of_widehat_f_a}.

Finally, the two bottom panels of \Figure~\ref{fig:bounds} display the same bounds except that the kernel eigenvalues now decay exponentially to zero as $\lambda_e = \exp(-0.1 e)$. Exponentially decaying eigenvalues are typical for Gaussian kernels, and therefore of practical relevance. The shapes of the curves change but the same comments hold true.

\subsection{Asymptotic behaviors of the estimators and of the bounds}
\label{ssec:asymptotic_behaviours_of_the_bounds}

Now, we start investigating the asymptotic properties of our estimators considering a situation where their dimension $E$ is fixed while the number of measurements $M$ grows to infinity. The next result then shows that $\widehat{f}_{A}$ and $\widehat{f}_{B}$ asymptotically reach the lower bound~\eqref{equ:Err_lower_bound}.
\begin{theorem}
	Given the assumptions in \Section~\ref{ssec:measurements_model} and \Assumption~\ref{ass:eigenfunctions_are_uniformly_bounded},
	\begin{equation}
		\displaystyle
		\lim_{M \rightarrow + \infty}
		\mathrm{Err}_{A}
		=
		\sum_{e = E + 1}^{+\infty} \lambda_{e}
		\qquad
		\text{in probability}
	\label{equ:f_A_is_asymptotically_efficient}
	\end{equation}
	\begin{equation}
		\displaystyle
		\lim_{M \rightarrow + \infty}
		\mathrm{Err}_{B}
		=
		\sum_{e = E + 1}^{+\infty} \lambda_{e}
		\qquad
		\text{in probability.}
	\label{equ:f_B_is_asymptotically_efficient}
	\end{equation}
	\label{thm:f_A_and_f_B_are_asymptotically_efficient}
\end{theorem}

We now discuss the statistical consistency of our estimators.
In this case, the conditions under which $\widehat{f}_{A}$ and $\widehat{f}_{B}$ converge to $f$ as both $E$ and $M$ grow to infinity
are different, as illustrated in the following two results. 
\begin{theorem}
	Given the assumptions in \Section~\ref{ssec:measurements_model} and \Assumption~\ref{ass:eigenfunctions_are_uniformly_bounded},
	\begin{equation}
		\lim_{M \rightarrow +\infty} \ \ \lim_{E \rightarrow +\infty}
		\mathrm{Err}_{A} 
		=
		0
		\qquad
		\text{in probability}
	\label{equ:f_A_is_consistent}
	\end{equation} 
	\label{thm:f_A_consistent}
\end{theorem}
\begin{theorem}
	Let $E=E(M)$ such that
	\begin{equation}
		E(M) \log E^\delta(M)
		\leq M^\delta, \ \ \ \lim_{M\to +\infty} E(M)=+\infty
		\label{equ:implicit_E_for_consistency_first_version}
	\end{equation}
	for some $\delta\in(0,1)$. Given the assumptions in \Section~\ref{ssec:measurements_model} and \Assumption~\ref{ass:eigenfunctions_are_uniformly_bounded}, then	
	\begin{equation}
		\lim_{\scriptscriptstyle
				\begin{array}{c}
					M \rightarrow +\infty \\
					E = E(M)
				\end{array}}
		\mathrm{Err}_{B} 
		=
		0
		\qquad
		\text{in probability}
		\label{equ:f_B_is_consistent}
	\end{equation} 
	\label{thm:f_A_and_f_B_are_consistent}
\end{theorem}
\begin{remark}
The sufficient condition required in the theorem in terms of the growth rate of $E(M)$ as a function of $M$ is tight according to the Chernoff's bound. In fact, our requirement is that $M$ grows up a bit more slowly w.r.t.\ the relationship $E \log E=M$. Now, assume instead that $E \log E=M$, i.e., $\delta=1$ and fix any rule such that $\varepsilon \rightarrow 1$ and $\alpha \rightarrow 0$. Recall that 
$$
1 - \varepsilon + \varepsilon \log(\varepsilon)
		\geq
		\frac{Ek}{M} \log\left(\frac{2 E}{\alpha}\right)
$$
	must be satisfied. Asymptotically, the lhs tends to $0^+$ while the second term becomes $k-\frac{Ek}{M} \log(\alpha/2)$ and is larger than $k$ when $\alpha$ is sufficiently close to zero. One would thus need $0 \geq k$ but this is not possible. Also note that the previous theorem implies that any sublinear power growth of $E(M)=M^a$, for any $a\in(0,1)$, satisfies the consistency condition, which can be readily verified by choosing $\delta=\frac{1+a}{2}$. 
\end{remark}

The consistency properties of $\widehat{f}_{A}$ and $\widehat{f}_{B}$ are thus remarkably different. For what regards $\widehat{f}_{A}$, as $M$ goes to infinity its consistency is guaranteed without any control on the growth rate of the dimension $E$. Indeed, as $E$ increases such estimator can approximate arbitrarily well the optimal $\widehat{f}_{\MAP}$. This agrees with what already discussed in the previous subsection: when using $\widehat{f}_{A}$ it is convenient for the network to use a dimension $E$ as large as possible, just compatible with its communication constraints. Differently, the estimator $\widehat{f}_{B}$ is instead consistent only if $M$ augments sufficiently faster than $E$. 

\section{Distributed tuning of the regularization parameter}
\label{sec:distributed_tuning_of_the_hyperparameters_of_widehat_f_a} 

The statistical bounds obtained in the previous section quantify the performance of $\widehat{f}_{A}$ and $\widehat{f}_{B}$ assuming that the prior function model is correct. Beyond their theoretical interest, in real applications these bounds can give useful guidelines to select the amount of information that agents need to exchange. However, the covariance $K$ is often defined only except for a scalar factor $\gamma$. In addition, the prior is never perfect and the tuning of $\gamma$ could also hinder possible undermodeling. So, in place of~\eqref{equ:distribution_of_the_estimand}, in practical applications it is beneficial to consider
\begin{equation}
	f \sim \GaussianDistribution{0}{\gamma^{-1} K}
	\label{equ:distribution_of_the_estimand2}
\end{equation}
with $\gamma$ to be estimated from the observed noisy outputs and related input locations. Furthermore, when $\widehat{f}_{B}$ is considered, it has been shown that also the parameter $E$ plays an important role since, for a fixed number of samples $M$, its performance degrades if $E$ is too small or too large. Hence, it could be desirable to adjust also the number of eigenfunctions forming the estimate after seeing the data.

In the following we will follow the \ac{SURE} approach for tuning the free parameters. Although alternative approaches are possible, such as cross validation and marginal likelihood optimization, we will see that \ac{SURE} has the advantage to require less communication and computation processing, and also to be suitable for distributed implementations. We start by reporting a result obtained through a simple generalization of the arguments in \cite{Hastie09}[Section 7.4].
 
\begin{theorem}
Let $\bm{\eta}$ be a deterministic unknown parameter vector. Assume that the measurements model is
$$
\bm{z} = \bm{\eta}+ \bm{e}
$$
and consider also future measurements 
$$
\bm{z}^{\ast} = \bm{\eta} + \bm{e}^{\ast}
$$ 
where the noises $\bm{e}$ and $\bm{e}^{\ast}$ are uncorrelated, zero mean with covariance $\Sigma$. Then, given the linear estimator $\widehat{\bm{z}}=S \bm{z}$, an unbiased estimator of the risk $\ExpectationOf{\| \bm{z}^{\ast} -\widehat{\bm{z}}\|^2}$ is given by:
\begin{equation}
	\left\|
		\bm{z} - \widehat{\bm{z}}
	\right\|^{2}
	+
	2 \TraceOf{S \Sigma} .
	\label{equ:sure_general}
\end{equation}
\end{theorem}
\noindent The quantity $\TraceOf{S \Sigma}$ entering the second part of the objective \eqref{equ:sure_general} is connected to the concept of \emph{equivalent degrees of freedom} \cite{Mackay:92,DeNicolao1}.

In what follows, we assume that $\gamma$ is unknown but belongs to the finite set $\Gamma$ which is known in advance to the network. In addition, let us assume that the estimation step has been performed adopting a certain value $E$. Hence, if $\widehat{f}_{A}$ has been used, each agent has stored $\frac{G^{T} G}{M}$ and $\frac{G^{T}}{M}y$ so that, letting
$$
H_{A}(\gamma)
\DefinedAs
\left( \frac{G^{T} G}{M} + \frac{\gamma \sigma^{2}_{\nu}}{ M} \Lambda^{-1}_{E} \right)^{-1} \frac{G^{T}}{M},
$$
it can compute $H_Ay$ for any $\gamma \in \Gamma$.

 If $\widehat{f}_{B}$ has been adopted, then also the optimal number of eigenbases $E'$ has to be found within the set $E' \in \Omega$. In this case, each agent knows only $\frac{G^{T}}{M}y$ and, letting 
\begin{equation}
	H_{B}(\gamma,E')
	\DefinedAs
	\mathcal{I}_{E'} \left( I + \frac{\gamma \sigma^{2}_{\nu}}{M} \Lambda^{-1}_{E} \right)^{-1} \frac{G^{T}}{M},
	\label{equ:definition_of_H_B_gamma_E_prime}
\end{equation}
where
\begin{equation}
	\mathcal{I}_{E'}
	\DefinedAs
	\begin{bmatrix}
		I_{E'} & \\
		& \bm{0}_{E - E'}\\
	\end{bmatrix},
	%
	%
	%
	\label{equ:definition_of_mathcal_I_of_E}
\end{equation}
it can compute $H_By$ for any $\gamma \in \Gamma$ and integer $E' \in \Omega$.

\subsection{Distributed \ac{SURE} for $\widehat{f}_{A}$: tuning of $\gamma$}
\label{ssec:distributed_sure_for_widehat_f_a}

The first strategy is suited for $\widehat{f}_{A}$. Surprisingly, we will see that the tuning of $\gamma$ can be performed by the network using only local operations, without the need of performing any additional consensus operation. Now, let us reconsider our measurements model 
$$
\bm{y} = G \bm{a} + Z \bm{b} + \bm{\nu}
$$
where $\bm{a}$ is $E$-dimensional. Hereby, we break away from the assumptions on prior correctness by thinking of $G \bm{a} + Z \bm{b}$ as a deterministic vector. It thus corresponds to the deterministic function $f$ sampled on the realizations of the input locations.


We then create a (projected) measurement model via pre-multiplication by $G^{T}/M$, i.e.,
\begin{equation}
	\underbrace{\frac{G^{T} \bm{y}}{M}}_{\bm{z}} 
	=
	\underbrace{\frac{G^{T} G}{M} \bm{a}
	+
	\frac{G^{T} Z}{M} \bm{b}}_{\bm{\eta}}
	+
	\underbrace{\frac{G^{T} \bm{\nu}}{M}}_{\bm{e}}
	\label{equ:measurement_model_alternative_compact_SURE} 
\end{equation}
where the correspondences with the key quantities defining the risk estimator \eqref{equ:sure_general} have been pointed out. From such definitions, we also obtain $\widehat{\bm{z}} = \frac{G^{T} G}{M} H_{A} \bm{y} = S \bm{z}$ where 
\begin{equation}
	S
	\DefinedAs
	\frac{G^{T} G}{M} \left( \frac{G^{T} G}{M} + \frac{\gamma \sigma^{2}_{\nu}}{ M} \Lambda^{-1}_{E} \right)^{-1} 
	\label{equ:S_for_distributed_SURE}
\end{equation}
and
\begin{equation}
	\Sigma = \sigma^{2}_{\nu} \frac{G^{T} G}{M^{2}}. 
	\label{equ:Sigma_for_distributed_SURE}
\end{equation}
Recall that the matrix $V=\frac{G^{T} G}{M}=\frac{1}{M}\sum_{m=1}^M G_m^T G_m$ and the vector $\bm{z} = \frac{1}{M}\sum_{m=1}^M G_m^T y_m$ have been already computed by each agent via a distributed consensus algorithm~\cite{garin2010survey} to implement $\widehat{f}_{A}$. 
Then, since the network cardinality $M$ is known, each agent can tune $\gamma$ by optimizing the \ac{SURE} score \eqref{equ:sure_general} connected with the prediction risk on the future data $\bm{z}^{\ast} =\frac{G^{T} \bm{y}^{\ast}}{M}$, i.e.,
\begin{equation}
	\widehat \gamma_A
	=
	\arg \min_{\gamma\in \Gamma} J_{A} \left( \gamma \right)
	\label{equ:SURE_A}
\end{equation}
with
\begin{equation}
	\begin{array}{rl}
	J_{A} \left( \gamma \right)
	\DefinedAs &
	 \left\| (I-S) \bm{z} \right\|^2 + 2 \TraceOf{S \Sigma}\\
	= &
	\displaystyle
	\left\|
		\frac{\gamma \sigma^{2}_{\nu}}{M}
		\left( V\Lambda_{E} + \frac{\gamma \sigma^{2}_{\nu}}{ M} I \right)^{-1} \bm{z}
	\right\|^2 + \\
	&
	\displaystyle
	+ \frac{2 \sigma^{2}_{\nu}}{M} \TraceOf{V^2 \left( V + \frac{\gamma \sigma^{2}_{\nu}}{ M} \Lambda^{-1}_{E} \right)^{-1}}. \\
	\label{equ:definition-of-J-A}
	\end{array}
\end{equation}
%

To understand the rationale underlying this strategy we have just to consider that the novel process~\eqref{equ:measurement_model_alternative_compact_SURE} is formed by $E$ measurements, each corresponding to the projection of the original ones on the space of the sampled eigenfunctions $\left[ \phi_{e}(x_{1}) \; \cdots \; \phi_{e}(x_{M}) \right]$. For large $M$, the quantity $G^{T} Z {\bm{b}}$ vanishes so that $\bm{\eta} \approx \bm{a}$. This means that the \ac{SURE} score becomes an unbiased estimator of those signal components which are expected to capture the most part of the energy. 

\begin{remark}
	Based on the previous analysis, it is straightforward to observe that the \ac{SURE} strategy described above is not suited for $\widehat{f}_{B}$. In fact, it requires each agent to know $\frac{G^{T} G}{M}$. But if this quantity were known, each agent could implement $\widehat{f}_{A}$, an estimator that has more favorable features than $\widehat{f}_{B}$. 
\end{remark}

\subsection{Distributed \ac{SURE} for $\widehat{f}_{B}$: tuning of $E$ and $\gamma$}
\label{ssec:distributed_sure_for_widehat_f_b}

The second strategy is designed for $\widehat{f}_{B}$. It tunes $\gamma \in \Gamma$ and $E' \in \Omega$ just using an additional average consensus on a vector of size $E \dim(\Omega)\dim(\Gamma)$. Our starting point is still \eqref{equ:measurement_model_alternative_compact_SURE}, i.e., the $E$-dimensional projected measurement space, where $\bm{z} = \frac{G^{T} \bm{y}}{M} = \frac{1}{M}\sum_{m=1}^M G_m^Ty_i \in \Reals^E$ has been computed to implement $\widehat{f}_{B}$ via a standard distributed consensus algorithm and is therefore known to each agent. Let us define $\widehat{a}(\gamma,E')=H_{B}(\gamma,E') \bm{y} \in \Reals^E$. Clearly $\widehat{a}(\gamma,E')$ for $E'<E$ is simply the truncated version of $\widehat{a}(\gamma,E)$ where the last $E-E'$ components are set to zero. Moreover, the vectors $\widehat{a}(\gamma,E')$ can be independently computed by each agent for each value of $\gamma\in \Gamma,E'\in \Omega$ once $\bm{z}$ is available. The output prediction can be written as 
$$
	\widehat{\bm{z}}(\gamma,E')
	=
	\frac{G^{T} G}{M} \widehat{a}(\gamma,E')= \frac{1}{M}\sum_{m=1}^MG_m^TG_m \widehat{a}(\gamma,E')\in\Reals^E.
$$
Hence, each agent can compute the vectors $\widehat{\bm{z}}(\gamma,E')$ for each $\gamma \in \Gamma$ and $E' \in \Omega$ by running an additional consensus of size $O(\dim(\Omega) \dim(\Gamma) E)$. As so, the first part of the \ac{SURE} score $\left\| \bm{z} - \widehat{\bm{z}}(\gamma,E') \right\|^{2}$ can be readily computed by each agent.

As for the second part of \ac{SURE} related to the equivalent degrees of freedom, we need to compute $\TraceOf{S(\gamma,E') \Sigma}$ where
\begin{equation}
	S(\gamma,E')
	= 
	\frac{G^{T} G}{M} \mathcal{I}_{E'} \left( I + \frac{\gamma \sigma^{2}_{\nu}}{M} \Lambda^{-1}_{E} \right)^{-1}, \ \ \ \Sigma = \frac{\sigma_{\nu}^2}{M}\frac{G^TG}{M}.
	\label{equ:S_for_distributed_SURE2}
\end{equation}
Obviously, this would not make too much sense in the context of $\widehat{f}_B$ since the computation of $V=\frac{G^TG}{M}$ would allow us to compute $\widehat{f}_A$ which has better performance anyways. Therefore we will approximate such matrix $V$ (similarly to what we did to obtain $\widehat{f}_B$) by replacing it with an identity matrix. This corresponds to use a sort of \emph{expected equivalent degrees of freedom}:
$$
\begin{array}{lll}
	\TraceOf{ \rule{0cm}{0.4cm} S(\gamma,E') \Sigma}
	&\approx&
	\displaystyle
	\frac{\sigma^{2}_{\nu}}{M}\TraceOf{\mathcal{I}_{E'} \left( I + \frac{\gamma \sigma^{2}_{\nu}}{M} \Lambda^{-1}_{E} \right)^{-1}} \\
	&=&
	\displaystyle
	\frac{\sigma^{2}_{\nu}}{M} 
	\sum_{e=1}^{E'} \frac{\lambda_e}{\lambda_e + \gamma \sigma^{2}_{\nu}/M} . \\
\end{array}
$$
The optimal tuning of the parameter is then obtained as
\begin{equation}
	\left( \widehat \gamma_B, \widehat{E}_B \right)
	=
	\operatorname*{\mathrm{argmin}}_{\gamma \in \Gamma,E'\in \Omega}
	J_{B} \left( \gamma, E' \right)
	\label{equ:SURE_B}
\end{equation}
with
\begin{equation}
	J_{B} \left( \gamma, E' \right)
	\DefinedAs
	\left\| \bm{z} - \widehat{\bm{z}}(\gamma,E') \right\|^{2}
	+
	2
	\frac{\sigma^{2}_{\nu}}{M}
	\sum_{e=1}^{E'} \frac{\lambda_e}{\lambda_e+ \gamma \sigma^{2}_{\nu}/M} .
	\label{equ:definition-of-J-B}
\end{equation}
Note that this strategy for tuning $\widehat{f}_B$ is more efficient from a communication and computational point of view than $\widehat{f}_A$ only if $\dim(\Omega) \dim(\Gamma)<E$.

\subsection{Practical implementation issues}
\label{InPractice}

We now illustrate how to implement the proposed distributed estimators, also in connection with the properties of the \ac{SURE} tuning strategies described above. We discuss first the use and the derivation of the \ac{KL} expansion and then how $f$ can be estimated in a distributed way also adopting basis functions different from the kernel eigenfunctions. All the code developed for implementing the algorithms below is publicly available in the repository \url{github.com/damianovar/Gaussian-regression-via-finite-dimensional-approximations}.

\subsubsection{Computing the \ac{KL} expansions}
\label{ssec:explicit_kl_expansions}

Assume that the prior on $f$ is correct and that the input locations distribution $\mu$ is known. Then, according to \Theorem~\ref{thm:f_A_and_f_B_are_asymptotically_efficient}, at least for large data set size $M$, the use of the eigenfunctions in \eqref{equ:eigenexpansion_of_K} is statistically optimal. Obtaining the kernel expansion in closed form is in general difficult but important exceptions are the popular spline and Gaussian kernel. In particular, for uniform $\mu$ the expansion of the linear and cubic smoothing spline kernel is reported in~\cite{Bell2004}. For Gaussian $\mu$ on the real line, the Gaussian kernel expansion is given via Hermite polynomials, as reported in~\cite{zhu_et_al__1998__gaussian_regression_and_optimal_finite_dimensional_linear_models}[Section 4]. Such result then immediately generalizes to multi-dimensional domains: if $\mu(\cdot)$ and $K(\cdot, \cdot)$ are tensor products of one-dimensional distributions and kernels, respectively, the expansion involves tensor products of the one-dimensional eigenfunctions.

Assume then that the kernel expansion is not available in closed form. It is worth pointing out that in many relevant distributed problems the dimension of the function domain $\mathcal{X}$ is limited to 2 or 3, and this makes the numerical determination of the eigenfunctions and eigenvalues viable. More specifically, let $\{\widetilde{x}_{e}\}_{e=1}^q$ be independent samples from $\mu$, and let $\bm{K}$ be the $q \times q$ kernel matrix whose $(i,j)$-entry is \vspace{0.1cm} 
\begin{equation}
	[\bm{K}]_{ij}
	=
	K \left( \widetilde{x}_i, \widetilde{x}_j \right),
	\quad i=1,\ldots,q, \quad j=1,\ldots,q
	\vspace{0.1cm}
	\label{equ:KernelMatrix}
\end{equation}
Then, according to~\cite{de1999consistent}[Lemma 9 and Corollary 10], the eigenvalues and (normalized) eigenvectors from the \ac{SVD} of $\bm{K}$ converge to the eigenvalues and eigenfunctions of $K(\cdot, \cdot)$ as $q \rightarrow +\infty$. Hence, the agents can be equipped with arbitrarily accurate approximations of the \ac{KL} expansion.

\subsubsection{Generic basis functions: Kernel sections}
\label{ssec:using_kernel_sections}

As discussed above, in some circumstances the kernel eigenfunctions could be not available in closed form, or have a complex functional form that makes storing them in the agents' memory unpractical. In such cases, one would rather use basis functions which admit simple closed-form expressions, possibly also non orthonormal. Even if the bounds developed in \Section~\ref{sec:statistical_characterizations_of_widehat_f_a_and_widehat_f_b} cannot be used anymore, we will see that the SURE strategies for hyperparameters tuning generalize well also to this situation.

We limit our discussion to the use of the kernel sections as basis (an important case also in view of their connections with the representer theorem~\cite{Kimeldorf70,Scholkopf01}). This basis is associated to a set $\{\widetilde{x}_{e}\}_{e=1}^E$ of input locations\footnote{As in \eqref{equ:KernelMatrix}, the set $\{\widetilde{x}_{e}\}_{e=1}^q$ is available a priori and has not to be confused with the input locations $\{x_{m}\}_{m=1}^M$ then visited by the agents.} which could be drawn from $\mu$ or selected in a deterministic way to cover sufficiently well $\mathcal{X}$. We then define our $E$ basis functions as
\begin{equation}
	\phi_{1}(\cdot)
	=
	K \left( \widetilde{x}_1, \cdot \right)
	\qquad \ldots \qquad
	\phi_{E}(\cdot)
	=
	K \left( \widetilde{x}_E, \cdot \right).
	\label{equ:KernelSections}
\end{equation}
Using the kernel sections in the decomposition~\eqref{equ:estimand_model}, we can think of $f_{a}$ as
\begin{equation}
	f_{a}(x) = \sum_{e = 1}^{E} a_{e} K \left( \widetilde{x}_e, x \right)
\end{equation}
where the vector $\bm{a} \DefinedAs \left[ a_{1}, \ldots, a_{E} \right]^{T}$ is now zero-mean Gaussian with covariance proportional to the inverse of the kernel matrix 
\begin{equation}
	[\bm{K}]_{ij}
	=
	K \left( \widetilde{x}_i, \widetilde{x}_j \right),
	\quad i=1,\ldots,E, \quad j=1,\ldots,E
	\vspace{0.1cm}
\end{equation}
i.e.,
$$
\bm{a} \sim \GaussianDistribution{0}{\gamma^{-1} \bm{K}^{-1}}.
$$
In fact, if the prior were correct, this would indeed correspond to see $f_{a}$ sampled on $\{\widetilde{x}_{e}\}_{e=1}^E$ as zero-mean Gaussian with covariance $\gamma^{-1} \bm{K}^{-1}$.

Since the kernel sections are generally not orthonormal w.r.t.\ $\mu$, even if $M \rightarrow \infty$ the projected measurements $\frac{G^{T} \bm{y}}{M}$ do not converge to the expansion coefficients $a_{e}$. However, these can be still used to tune the regularization parameters. In particular, for what concerns $\widehat{f}_{A}$, the distributed \ac{SURE} estimator introduced in \Section~\ref{ssec:distributed_sure_for_widehat_f_a} can estimate $f$ and $\gamma$ with a single consensus just replacing $\Lambda^{-1}_{E}$ with $\bm{K}$. Thus, this estimator does not even need the knowledge of $\mu$ and the agents can implement it once they know the function $K\left( \cdot, \cdot \right)$ and the expansion grid $\{\widetilde{x}_{e}\}_{e=1}^E$. The estimator $\widehat{f}_{A}$ is thus given by:
$$ \widehat{\bm{a}}(\gamma)
	\DefinedAs
	\left( \frac{G^{T} G}{M} + \frac{\sigma^{2}_{\nu}}{M} \bm{K} \right)^{-1} \frac{G^{T} \bm{y}}{M} 
 $$
$$S(\gamma)
	\DefinedAs
	\frac{G^{T} G}{M} \left( \frac{G^{T} G}{M} + \frac{\gamma \sigma^{2}_{\nu}}{ M} \bm{K} \right)^{-1} $$

Consider now the implementation of $\widehat{f}_{B}$ through the kernel sections with the estimator defined by the set of potential $E' \in \Omega$. In particular, for the sake of simplicity, assume that each $E'$ is associated to the kernel sections induced by the first $E'$ input locations in the (ordered) set $\{\widetilde{x}_{e}\}_{e=1}^E$. Given a generic matrix $A$, the submatrix obtained by retaining its first $E'$ rows and columns is denoted by $[A]_{E'}$. Assume moreover that the same notation applies to vectors to retain only their first $E'$ elements. Then, the same \ac{SURE} strategy developed in \Section~\ref{ssec:distributed_sure_for_widehat_f_b} can be adopted by setting
\begin{equation}
	\widehat{\bm{a}} \left( \gamma, E' \right)
	=
	\begin{bmatrix}
		I_{E'} \\
		\bm{0} \\
	\end{bmatrix}
	\left(
		\left[\mathbb{E} \frac{G^{T} G}{M} \right]_{E'}
		+
		\frac{\gamma \sigma^{2}_{\nu}}{M} [\bm{K}]_{E'}
	\right)^{-1} \!\!
	\left[ \frac{G^{T} \bm{y}}{M} \right]_{E'}
	\label{equ:widehat_a_gamma_E_prime_with_kernel_sections}
\end{equation}
and
\begin{equation}
	S(\gamma,E')
	=
	\left[\mathbb{E} \frac{G^{T} G}{M} \right]_{E'}
	\begin{bmatrix}
		I_{E'} \\
		\bm{0} \\
	\end{bmatrix}
	\left(
		\left[\mathbb{E} \frac{G^{T} G}{M} \right]_{E'}
		+
		\frac{\gamma \sigma^{2}_{\nu}}{M} [\bm{K}]_{E'}
	\right)^{-1}
	\vspace{0.2cm} 
	\label{equ:S_gamma_E_prime_with_kernel_sections}
\end{equation}
where, instead of using $\mathcal{I}_{E'}$ defined in~\eqref{equ:definition_of_mathcal_I_of_E}, we use $I_{E'}$ and $\bm{0} \in \Reals^{(E - E') \times E'}$ to account for the non-diagonal nature of the matrices now at stake, with the trace of $S$ given by the sum of the its $(i,i)$ entries with $i=1,\ldots,E'$, and where $\ExpectationOf{\frac{G^{T} G}{M}}$ substitutes $I$ in~\eqref{equ:definition_of_H_B_gamma_E_prime}, since the kernel sections are generally not orthonormal. The exact expectation of $\ExpectationOf{\frac{G^{T} G}{M}}$ can be explicitly computed in some special cases as in the example below, or can be approximated via its sampled version, i.e., $\ExpectationOf{\frac{G^{T} G}{M}}=\ExpectationOf{\sum_{m=1}^M\frac{G_m^{T} G_m}{M}} \approx \frac{1}{E}\sum_{e=1}^E\frac{\widetilde{G}_e^{T} \widetilde{G}_e}{E}$ where $\widetilde{G}_e$ are computed on the input locations $\{\widetilde{x}_{e}\}_{e=1}^E$ that shall be used for computing this empirical expectation, not to be confused with the set of a-posteriori input locations $\{{x}_{m}\}_{m=1}^M$ used in the actual coefficients estimation step.

\begin{example}
	An interesting case, relevant for many applications, arises when one wants to use the Gaussian kernel and its kernel sections as basis with $\mu$ a mixture of Gaussians. In this case $\ExpectationOf{\frac{G^{T} G}{M}}$ can be obtained in closed form. In fact, consider first a scalar scenario, i.e., $x\in\mathbb{R}$, with a mixture of Gaussians made of a single component:\vspace{0.2cm}
$$
	K(x,x') = \exp \left( -\frac{(x-x')^2}{\eta} \right),
	\;
	\mu \sim \mathcal{N}(\mu_0,a^2).
	\vspace{0.2cm}
$$
After simple computations, one finds that the $(e,e')$-entry of $\ExpectationOf{\frac{G^{T} G}{M}}$ is
\begin{equation}
	\begin{array}{l}
		\displaystyle
		\int_{-\infty}^{+\infty}
		\frac
		{\exp \left(
			-\frac{(x-x_e)^2}{\eta}
			-\frac{(x-x_{e'})^2}{\eta}
			-\frac{(x-\mu_0)^2}{2a^2}
		\right)}
		{\sqrt{2 \pi}a} dx =
		\vspace{0.2cm} \\
		\displaystyle
		\qquad\qquad =
		\frac{\sqrt{\eta}}{\sqrt{\eta+4a^2}}
		\exp \left( -\frac{\star}{\eta^2+4\eta a^2} \right)
	\end{array}
	\vspace{0.2cm}
\end{equation}
with $\star = \eta \left( x_e^2-2\mu_0 x_e+x_{e'}^2-2\mu_0 x_{e'}+2\mu^2 \right) + 2a^2 \left( x_e-x_{e'} \right)^2$. In the multivariate case, assume that $K(x,x') = e^{-\frac{\|x-x'\|^2}{\eta}}$ while $\mu$ is given by tensor products of one-dimensional Gaussian densities. Then, the result is still available in closed form: $\ExpectationOf{\frac{G^{T} G}{M}}$ corresponds to convex combinations of Hadamard products of the matrices obtained in the scalar case.
\end{example}

\subsubsection{Generic basis functions: Nystr{\"o}m method}
\label{sssec:application_of_the_distributed_tuning_algorithms_to_the_nystrom_method}

The analysis provided in the previous section can also be applied to the popular Nystr{\"o}m method \cite{Drineas:2005,Yang:12}. The idea is to find a basis for $f$ of dimension $E$ which has almost the same performance of the basis composed of $q \gg E$ kernel sections with $q$ an arbitrary number as in \Section~\ref{ssec:explicit_kl_expansions}. More specifically, let $\{\widetilde{x}_{n}\}_{n=1}^q$ be $q$ input location defined a-priori\footnote The Nystr{\"o}m method closely resembles the eigenfunctions/eigenvalues numerical computation method presented in Section~\ref{ssec:explicit_kl_expansions}, the difference being that in Nystr{\"o}m the parameter $q$ is in general not extremely large and the input locations $\{\widetilde{x}_{n}\}_{n=1}^q$ are not generated by $\mu$ but are randomly extracted from the training set. from $\mu$, and consider both the corresponding kernel matrix $\bm{K} \in \Reals^{q \times q}$ defined in~\eqref{equ:KernelMatrix} and its \ac{SVD} decomposition $\bm{K} = V D_q V^T$ with $V \DefinedAs [v_1, \ldots, v_q]$ the orthonormal eigenvectors of $\bm{K}$ and $D_q$ the diagonal matrix formed by the corresponding eigenvalues of $\bm{K}$ sorted in non-increasing order. If $V_E \DefinedAs [v_1, \ldots, v_E]$ and $D_E \in \Reals^{E\times E}$ is the diagonal matrix with the first $E$ sorted eigenvalues of $\bm{K}$, then $\bm{K}_E = V_E D_E V_E^T$ is the best rank-$E$ approximation of $\bm{K}$. The a priori basis 
$$ 
\phi_e(x) \DefinedAs \sum_{n=1}^q v_e(n) K(\widetilde{x}_{n},x), \quad e = 1,\ldots,E
\vspace{0.1cm} 
$$
with $v_e(n)$ the $n$-th element of the vector $v_{e}$, can then be used to define 
$$f_{a}(x) = \sum_{e = 1}^{E} a_{e} \phi_e(x)$$
where $\bm{a} \DefinedAs \left[ a_{1}, \ldots, a_{E} \right]^{T}$ is a zero-mean Gaussian vector with 
$$
	\bm{a}
	\sim
	\GaussianDistribution{0}{\gamma^{-1} \left( V_E^T \bm{K} V_E \right)^{-1}}
	=
	\GaussianDistribution{0}{\gamma^{-1} D_E^{-1}}.
$$
We can then use once again~\eqref{equ:definition_of_G_and_Z} to build $G$ using the $\phi_e$'s above, and exploit the same strategies considered in \Section~\ref{ssec:using_kernel_sections} just replacing $\bm{K}$ with $D_E$.

\subsection{Numerical study on synthetic data}
\label{SURE:test}

Let us consider the same data generators based on the spline and the exponentially decaying kernels described in \Section~\ref{exa:increasing_E_can_lead_to_a_performance_loss}. The unknown function has to be reconstructed from $M=10000$ measurements by $\widehat{f}_{A}$ and $\widehat{f}_{B}$. The errors are still the \acp{MSE} defined in \eqref{MSEA} and \eqref{MSEB} normalized by the prior variance (the same definition was used to build \Figure~\ref{fig:bounds}). The difference however is that our estimators now depend on unknown hyperparameters that need to be inferred from data. More specifically, when $\widehat{f}_{A}$ is adopted we fix $E=400$ and the regularization parameter is searched over a grid $\Gamma$ containing 50 logarithmically spaced values between $10^{-3}$ and $10^3$. When using $\widehat{f}_{B}$ the grid $\Gamma$ contains only the three values $\{10^{-3},0,10^3\}$ while $E$ is estimated from data over $\Omega=\{1,5,10,20,50,100,200,300,400\}$. We still consider a Monte Carlo study of 1000 runs where at any run independent realizations of $f$, of the $M$ input locations and of the measurement noises are generated. Hyperparameters tuning is then performed by:
\begin{itemize}
	\item ``$\widehat{f}_{A}+\,$oracle'' and ``$\widehat{f}_{B}+\,$oracle'', where ``oracle'' indicates that these approaches know at any run the realization of $f$ (which is the object to estimate) and select exactly those hyperparameters that minimize the \ac{MSE} achievable by those two estimators. For instance, assume that $f=\sum_{e=1}^{\infty} \ a_e \phi_{e}$ is the realization of the function at a certain run. Let also $\widehat{a}(\gamma,E')$ denote the vector with the estimates of the first $E'$ coefficients $a_e$ returned by $\widehat{f}_{B}$. Then $\widehat{f}_{B}+\,$oracle determines the hyperparameters as
	$$
		\left( \widehat{\gamma}, \widehat{E} \right)
		\DefinedAs
		\arg\min_{\gamma \in \Gamma,E' \in \Omega} \
		\sum_{e=1}^{\infty} \
		\big( a_e - \widehat{a}_e(\gamma,E') \big)^2
	$$
	where $\widehat{a}_e(\gamma,E') \DefinedAs 0$ for $e>E'$. Thus, this estimator is not implementable in practice and provides the lower bound on the \acp{MSE} \eqref{MSEA} and \eqref{MSEB} achievable by the two estimators; 
\item ``$\widehat{f}_{A}+\,$SURE'' and ``$\widehat{f}_{B}+\,$SURE'', where the hyperparameters tuning step is performed following the \ac{SURE} approaches described in the previous subsection. Recall that ``$\widehat{f}_{A}+\,$SURE'' requires only a single consensus on a vector of size $O(E^2)$ to obtain simultaneously both the hyperparameters and function estimates, while ``$\widehat{f}_{B}+\,$SURE'' requires two consensus operations of size $O(E)$.
\end{itemize}

\Figure~\ref{fig:SURE} compares with a scatter-plot the various (normalized) \acp{MSE} obtained by the oracle- and \ac{SURE}-based approaches as a function of the Monte Carlo run. Remarkably, \ac{SURE}'s performance (dashed lines) is very close to that of the oracle (solid lines). When using ``$\widehat{f}_{B}+\,$SURE'' (right panels) the curves are in practice indistinguishable.

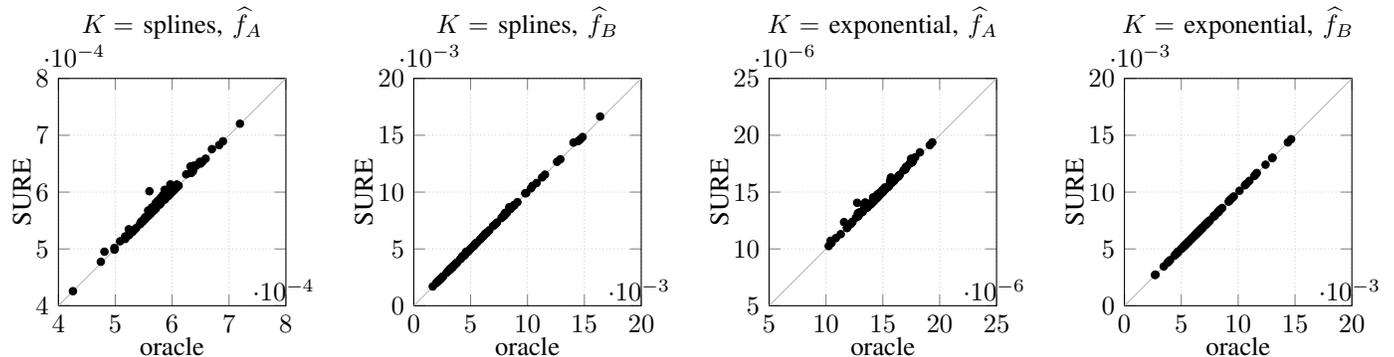
\begin{figure*}
	\centering
	\pgfplotsset
{
	AxisStyle/.style =
	{
		width					= 0.52\columnwidth,
		height					= 0.52\columnwidth,
		xlabel near ticks,
		ylabel near ticks,
		xlabel shift			= -5,
		ylabel shift			= -5,
		scaled ticks			= base 10:+4,
		every x tick scale label/.style = {at={(0.99,0.07)}},
		every y tick scale label/.style = {at={(0.07,1.1)}},
		title style		=
		{
			align		= center,
			yshift		= 0.2cm
		},
	},
	GridStyle/.style =
	{
		xmajorgrids,
		ymajorgrids,
		xminorgrids,
		minor grid style =
		{
			thin,
			densely dotted,
			black!20
		},
		major grid style =
		{
			thin,
			densely dotted,
			black!20
		},
	},
	BisectorStyle/.style =
	{
		solid,
		draw			= black!30!white,
		line width		= 0.01cm,
	},
	MarkStyle/.style =
	{
		scatter,
		scatter/use mapped color = {black},
		only marks,
		mark options	=
		{
			scale 		= 0.7,
			draw		= black,
			fill		= black,
			opacity		= 0.99
		},
		mark			= *,
	},
}
\def\myseparation{1.7cm}

\begin{tikzpicture}
	\begin{axis}
	[
		AxisStyle,
		GridStyle,
		title		= {$K = $ splines, $\widehat{f}_{A}$},
		name		= splinesBndA,
		xlabel		= {oracle},
		ylabel		= {SURE},
		xmin		= 0.0004, xmax		= 0.0008,
		ymin		= 0.0004, ymax		= 0.0008,
		scaled ticks= base 10:+4,
	]
		\addplot [MarkStyle] table [x = ErrAmc1, y = ErrAmc2]
		{data-ScatterPlotsSURE.txt};
		\addplot [BisectorStyle, domain = 0.0004:0.0008] {x};
	\end{axis}
	\begin{axis}
	[
		AxisStyle,
		GridStyle,
		title		= {$K = $ splines, $\widehat{f}_{B}$},
		xlabel		= {oracle},
		ylabel		= {SURE},
		name		= splinesBndB,
		at			= {($(splinesBndA.east)+(\myseparation,0)$)},
		anchor		= west,
		xmin		= 0.0000, xmax		= 0.0200,
		ymin		= 0.0000, ymax		= 0.0200,
		scaled ticks= base 10:+3,
	]
		\addplot [MarkStyle] table [x = ErrBmc1, y = ErrBmc2]
		{data-ScatterPlotsSURE.txt};
		\addplot [BisectorStyle, domain = 0.0000:0.0200] {x};
	\end{axis}
	\begin{axis}
	[
		AxisStyle,
		GridStyle,
		title		= {$K = $ exponential, $\widehat{f}_{A}$},
		xlabel		= {oracle},
		ylabel		= {SURE},
		name		= exponentialBndA,
		at			= {($(splinesBndB.east)+(\myseparation,0)$)},
		anchor		= west,
		xmin		= 0.000005, xmax		= 0.000025,
		ymin		= 0.000005, ymax		= 0.000025,
		scaled ticks= base 10:+6,
	]
		\addplot [MarkStyle] table [x = ErrAmcEx1, y = ErrAmcEx2]
		{data-ScatterPlotsSURE.txt};
		\addplot [BisectorStyle, domain = 0.000005:0.000025] {x};
	\end{axis}
	\begin{axis}
	[
		AxisStyle,
		GridStyle,
		title		= {$K = $ exponential, $\widehat{f}_{B}$},
		xlabel		= {oracle},
		ylabel		= {SURE},
		name		= exponentialBndB,
		at			= {($(exponentialBndA.east)+(\myseparation,0)$)},
		anchor		= west,
		xmin		= 0.00000, xmax		= 0.020,
		ymin		= 0.00000, ymax		= 0.020,
		scaled ticks= base 10:+3,
	]
		\addplot [MarkStyle] table [x = ErrBmcEx1, y = ErrBmcEx2]
		{data-ScatterPlotsSURE.txt};
		\addplot [BisectorStyle, domain = 0.0:0.02] {x};
	\end{axis}
\end{tikzpicture}
	\caption{Comparison of the MSE indexes obtained by the \ac{SURE}- and oracle-based strategies. Each circle corresponds to the result of a certain Monte Carlo run (the $x$-axis being associated to oracle-based estimators, and the $y$-axis to SURE-based ones). The fact that the circles groups are close to the bisector of the first quadrant indicates that the performance of \ac{SURE} is almost equivalent to that of the oracle.}
\label{fig:SURE}
\end{figure*}

The set of four Monte Carlo experiments have been also repeated adopting different data set sizes $M$. To synthesize \ac{SURE}'s performance with an index function only of $M$, let $\mathcal{S}_p \in [0,1]$ denote the ratio between the mean of the 400 errors obtained by the oracle and the \ac{SURE} strategies respectively for a certain value of $M$. Note that a value of $\mathcal{S}_p=1$ indicates that \ac{SURE} is performing as well as the oracle and that, for $M=10000$, $\mathcal{S}_p$ becomes the distillate of \Figure~\ref{fig:SURE}. \Table~\ref{Table1} reports $\mathcal{S}_p$ for $M=100,1000,10000$: one can see that the proposed hyperparameter estimation procedure behaves very nicely.

\begin{table}
\begin{center}
\begin{tabular}{cccc}\hline
Data set size & M = 100 & M = 1000 & M = 10000 \\
$\mathcal{S}_p$ & 0.93 & 0.987 & 0.99 \\
\hline \phantom{|}
\end{tabular} 
\end{center}
\caption{\ac{SURE}'s performance index $\mathcal{S}_p$ summarizing four Monte Carlo studies as a function of the number $M$ of available measurements. A value of $\mathcal{S}_p$ close to $1$ indicates that \ac{SURE}'s performance is close to that of the oracle. For $M=10000$, $\mathcal{S}_p$ represents a resume of the entire
\Figure~\ref{fig:SURE}.}
	\label{Table1}
\end{table}

\subsection{Numerical study on field data -- Colorado rain}
\label{sec:numerical_assessments_meteorology}

Let us now consider the reconstruction of monthly precipitations using data collected in Colorado in the years 1995-1997~\cite{coloradorain}. Many alternative solutions are available in the context of weather forecasts, but they are limited to centralized solutions, such as \cite{Gelfand:05,Datta:16}, for example, thus not suitable in our distributed framework. Measurements $y_{m}$ are a series of monthly average precipitations measured at 367 stations within the rectangular longitude/latitude region $[-109.5,-101] \times [36.5,41.5]$ remapped for convenience into the unitary square so that $x_{m} \in [0, 1]^{2}$ for every $m$.

We test the \ac{SURE} strategies~\eqref{equ:SURE_A} and~\eqref{equ:SURE_B}. When using $\widehat{f}_{A}$ we set $E = 20$ and $\Gamma$ to the grid containing 50 values logarithmically spaced between $10^{-5}$ and $10^{5}$. When adopting $\widehat{f}_{B}$ we use $\Gamma = 0$ and $\Omega = \left\{ 2, 4, \ldots, 20 \right\}$, i.e., consider only $E'$ as a regularization parameter. In both cases, we consider the Gaussian kernel
\begin{equation}
	K(x,x') = \exp \left( - 10 \| x - x' \|^{2} \right).
\end{equation}
The eigenfunctions are computed by assuming that the input locations are not know a-priori and extracted uniformly from the monitored region, i.e., $\mu(x)$ is a uniform distribution. We design a Monte Carlo study of 1000 runs where, at any run, we select randomly two months within the 1995-1997 dataset obtaining three different sets. The first one is a training set $\mathcal{D}_{\textrm{train}}$ of average precipitations obtained by selecting randomly and uniformly 2/3 of the measurements from the first selected month. The second is a test set $\mathcal{D}_{\textrm{test}}$ corresponding to the remaining 1/3 measurements from the first selected month. The last one is $\mathcal{D}_{\sigma_{\nu}^{2}}$ and contains measurements in the second selected month which are used to estimate the noise variance via least squares based on $E$ eigenfunctions. This corresponds to using $\widehat{f}_{A}$ with $\gamma=0$ obtaining as estimate of the noise variance 
\begin{equation}
	\widehat{\sigma}_{\nu}^{2}
	=
	\frac{1}{\dim(\mathcal{D}_{\sigma_{\nu}^{2}})-E}
	\sum_{m = 1}^{\dim(\mathcal{D}_{\sigma_{\nu}^{2}})} \left( \widehat{f}_{A} \left( x_{m} ; 0 \right) - y_{m} \right)^{2}
	\label{equ:estimate-variance-of-noise}
\end{equation}
where $\dim(\mathcal{D}_{\sigma_{\nu}^{2}})$ is the cardinality of $\mathcal{D}_{\sigma_{\nu}^{2}}$. Overall, this represents a situation where noise levels are determined by a centralized approach before running the estimators $\widehat{f}_{A}$ and $\widehat{f}_{B}$.

The following tuning strategies are used:
\begin{enumerate}
	\item ``$\widehat{f}_{A}+\,$oracle'' and ``$\widehat{f}_{B}+\,$oracle'', where ``oracle'' now indicates that these approaches can select those hyperparameters minimizing the following prediction errors on the test set
	\begin{equation}
		\textrm{RSS}_{A} \left( \gamma \right)
		\DefinedAs
		\frac{1}{\dim(\mathcal{D}_{\textrm{test}})}
		\sum_{m = 1}^{\dim(\mathcal{D}_{\textrm{test}})} \left( \widehat{f}_{A} \left( x_{m} ; \gamma \right) - y_{m} \right)^{2}
		\label{equ:RSS-FA-in-test-set}
	\end{equation}
	\begin{equation}
		\small
		\textrm{RSS}_{B} \left( \gamma, E' \right)
		\DefinedAs
		\frac{1}{\dim(\mathcal{D}_{\textrm{test}})}
		\sum_{m = 1}^{\dim(\mathcal{D}_{\textrm{test}})} \left( \widehat{f}_{B} \left( x_{m} ; \gamma, E' \right) - y_{m} \right)^{2}
		\label{equ:RSS-FB-in-test-set}
	\end{equation}
	where $\left( x_{m}, y_{m} \right)$ are all elements of $\mathcal{D}_{\textrm{test}}$. 
	Note that $\textrm{RSS}_{A}$ and $\textrm{RSS}_{B}$ can be seen as approximations of the \acp{MSE} \eqref{MSEA} and \eqref{MSEB} and that the oracle provides a lower bound on their values;
	\item ``$\widehat{f}_{A}+\,$SURE'' and ``$\widehat{f}_{B}+\,$SURE'', where the hyperparameters tuning step is performed minimizing the estimated risks $J_{A} \left( \gamma \right)$ and $J_{B} \left( \gamma, E' \right)$ defined in~\eqref{equ:SURE_A} and~\eqref{equ:SURE_B}.
\end{enumerate}

\Figure~\ref{fig:distributed_SURE_vs_oracle_Colorado_rain} compares with a scatter-plot the prediction errors \eqref{equ:RSS-FA-in-test-set} and~\eqref{equ:RSS-FB-in-test-set} achieved by the estimators after the 1000 Monte Carlo runs. The situation is not dissimilar from the case of synthetic data: as for the estimators ``$\widehat{f}_{A}+\,$SURE'' and ``$\widehat{f}_{B}+\,$SURE'', the performance of the \ac{SURE} strategies is close to that of the oracles.

\begin{figure}[!htbp]
	\centering
	\input{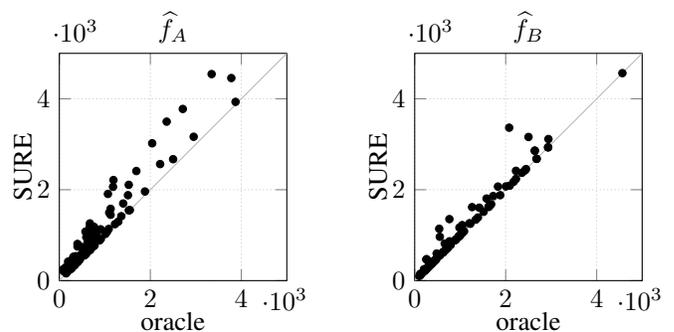}
	\caption{Comparison of the \ac{RSS} prediction error indexes obtained by the oracle- and \ac{SURE}-based strategies. Each opaque circle corresponds to the result of one of the 1000 Monte Carlo runs. The closer the circles are to the bisector of the first quadrant performance means that the closer the performance of that \ac{SURE}-based or Nystr{\"o}m-\ac{SURE} estimator is to the ones of the oracle-based estimator.}
	\label{fig:distributed_SURE_vs_oracle_Colorado_rain}
\end{figure}

Specifically considering the \ac{SURE}-based strategies, \Figure~\ref{fig:1D_Test_RSS_vs_SURE_risk_FA} also compares the estimated risks $J_{A} \left( \gamma \right)$ and $J_{B} \left( \gamma, E' \right)$ against the indexes $\textrm{RSS}_{A} \left( \gamma \right)$ and $\textrm{RSS}_{B} \left( \gamma, E' \right)$ in~\eqref{equ:RSS-FA-in-test-set} and~\eqref{equ:RSS-FB-in-test-set} in the first Monte Carlo run. The curves show that hyperparameters values have a major effect on the estimation performance and that our \ac{SURE} approach leads to a good regularization tuning. The related function estimates are visible in \Figure~\ref{fig:colorado-rain-regression}.

\begin{figure}[!htbp]
	\centering
	\pgfplotsset
{
	PlotStyle/.style =
	{
		width					= 0.55\columnwidth,
		height					= 0.55\columnwidth,
		xlabel near ticks,
		xlabel shift			= 1,
		ylabel					= {},
		ylabel near ticks,
		ymajorgrids,
		major grid style =
		{
			thin,
			densely dotted,
			black!20
		},
		x tick label style		= {anchor = north, inner ysep = 0.1cm},
		y tick label style		= {inner xsep = 0.1cm},
		legend plot pos			= left,
		legend columns			= 1,
		legend style			= {draw = none, fill = white},
		legend cell align		= left,
		inner xsep				= 0.0cm,
		inner ysep				= 0.0cm,
		legend style			=
		{
			nodes				= {font = \scriptsize},
			at					= {(0.5,1.00)},
			anchor				= south,
		}
	},
	RSSStyle/.style =
	{
		smooth,
		mark			= none,
		solid,
		draw			= black!30!white,
		line width		= 0.05cm,
	},
	FitStyle/.style =
	{
		smooth,
		mark			= none,
		solid,
		draw			= blue!70!white,
		line width		= 0.03cm,
	},
	RiskStyle/.style =
	{
		RSSStyle,
		dashed,
		draw			= black!80!white,
		line width		= 0.02cm,
	},
}

\begin{tikzpicture}
	\begin{axis}
	[
		PlotStyle,
		name		= fA,
		ymode		= log,			
		xmin		= -5,
		xmax		= 5,
		xlabel		= {$\log_{10} \left( \gamma \right)$},
	]
		\addplot [RSSStyle] table [x = log10PotentialGammas, y = afTestRSSsFA]
		{data-1D_Test_RSSs_vs_SURE_Risks_FA.txt};
		\addlegendentry{$\textrm{RSS}_{A} \left( \gamma \right)$};
		\addplot [RiskStyle] table [x = log10PotentialGammas, y = afSURERisksFAtimes]
		{data-1D_Test_RSSs_vs_SURE_Risks_FA.txt};
		\addlegendentry{$J_{A} \left( \gamma \right)$};
	\end{axis}
	\begin{axis}
	[
		PlotStyle,
		name		= fB,
		at			= {($(fA.east)+(1.2cm,0)$)},
		anchor		= west,
		ymode		= log,			
		xmin		= 0,
		xmax		= 20,
		xlabel		= {$E'$},
	]
		\addplot [RSSStyle] table [x = E, y = afTestRSSsFB]
		{data-1D_Test_RSSs_vs_SURE_Risks_FB.txt};
		\addlegendentry{$\textrm{RSS}_{B} \left( 0, E' \right)$};
		\addplot [RiskStyle] table [x = E, y = afSURERisksFBtimes]
		{data-1D_Test_RSSs_vs_SURE_Risks_FB.txt};
		\addlegendentry{$J_{B} \left( 0, E' \right)$};
	\end{axis}
\end{tikzpicture}
	\caption{Comparison of the predictive performance of the estimators $\widehat{f}_{A}$ (left panel) and $\widehat{f}_{B}$ (right panel) over the test set in \Figure~\ref{fig:colorado-rain-regression} for $\gamma \in \Gamma$ and $E' \in \Omega$ against the \ac{SURE} scores $J_{A} \left( \gamma \right)$ and $J_{B} \left( \gamma \right)$ in the first Monte Carlo run.}
	\label{fig:1D_Test_RSS_vs_SURE_risk_FA}
\end{figure}
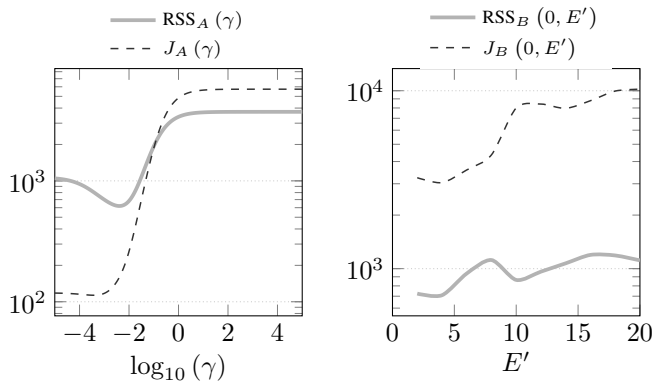

\begin{figure}[!htbp]
	\centering
	\includegraphics[width = 0.8\columnwidth]{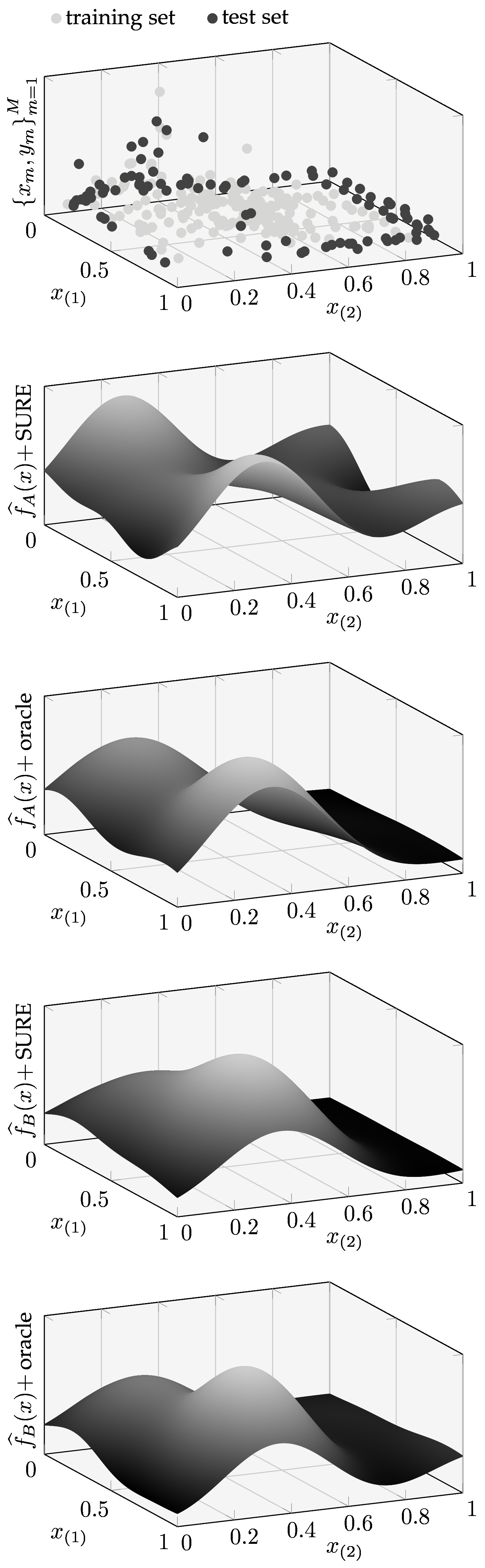}
	\caption{Visualization of the training and test sets (top panel, respectively 173 and 87 samples), and of the estimates returned by ``$\widehat{f}_{A}+\,$SURE'', ``$\widehat{f}_{A}+\,$oracle'', ``$\widehat{f}_{B}+\,$SURE'', and ``$\widehat{f}_{B}+\,$oracle'' in the first Monte Carlo run.}
	\label{fig:colorado-rain-regression}
\end{figure}

\subsection{Numerical study on field data -- UCI datasets}
\label{ssec:numerical_study_on_field_data__cpu_performance}

The second study is performed on two datasets from the public UCI repository, and is executed using the Nystr{\"o}m-based strategy described in \Section~\ref{sssec:application_of_the_distributed_tuning_algorithms_to_the_nystrom_method} to compute the basis functions for the estimators using all the input locations that define the training set. Our purpose is here to compare the proposed \ac{SURE}-based strategy for tuning the regularization parameters against an oracle that selects as the best regularization parameters that ones that give the best fit performance in the test set. As for the kernel, we consider a Gaussian kernel with an unitary variance (not accurately tuned, since the purpose of this section is more checking that the proposed \ac{SURE} strategy chooses the regularization parameters accurately rather than actually maximizing the generalization capabilities of the regression algorithms). As for the grid for tuning $\gamma$, we then consider the set $\Gamma = \{ 0, 10^{-5}, 10^{-4}, \ldots, 10^{2} \}$; as for $E'$, we consider $\Omega = \{ 1, 2, \ldots, 30 \}$. \Table~\ref{tab:numerical-study-on-UCI-datasets} summarizes then the obtained results, and numerically confirms the efficacy of the proposed parameters tuning strategies.
\def\FA{$\widehat{f}_{A}$}
\def\FB{$\widehat{f}_{B}$}
\begin{table}
\rowcolors
{2}					
{black!10!white}	
{}					
	\centering
	\begin{tabular}{l|cc|cc|}
				& \multicolumn{2}{c}{CCPP}	& \multicolumn{2}{c}{CPU}	\\
				& \FA		& \FB			& \FA		& \FB			\\
fit oracle		& 99.3		& 73.3			& 75.7		& 76.1			\\
fit SURE		& 99.3		& 73.3			& 71.2		& 66.5			\\
$\gamma$ oracle	& $10^{-3}$	& -				& $10^{-4}$	& -				\\
$\gamma$ SURE	& $10^{-3}$	& -				& $10^{-5}$	& -				\\
$E'$ oracle		& -			& 2.00			& -			& 10.00			\\
$E'$ SURE		& -			& 2.00			& -			& 27.00			\\
	\end{tabular} 
	\caption{Summary of the performance of the proposed parameters calibration strategies against oracles for different publicly available datasets. "CCPP" indicates the UCI \texttt{Combined Cycle Power Plant} regression dataset, with 9568 instances and a 4-dimensional input domain $\mathcal{X}$. "CPU" indicates the UCI \texttt{Computer hardware} regression dataset, with 209 instances and a 6-dimensional input domain $\mathcal{X}$ (corresponding to the quantitative features available in the dataset). For each dataset 1/3 of the data has been used for test purposes.}
\label{tab:numerical-study-on-UCI-datasets}
\end{table}

\section{Conclusions}
\label{Conclusions}

Distributed function estimation is an important problem where agents with limited computational, data storage and communication capabilities collect noisy measurements and have to reconstruct an unknown map in a collaborative way. In this context, we have studied Gaussian regression providing rigorous statistical bounds on the performance of two distributed estimators, also characterizing their asymptotic behavior. On the practical side, our study indicates how the dimension $E$ of the adopted estimator has to depend on the number of measurements $M$ collected by the agents to guarantee the desired statistical performance. The analysis clarifies merits and limitations of the two approaches also in function of the different amount of information exchange required to the network (linear or quadratic in $E$). We have also introduced novel distributed strategies which learn from data possibly unknown hyperparameters entering the estimators, and that do not necessarily require solving potentially numerically intensive eigenfunctions-eigenvalues decompositions of kernel functions. For the first time, to our knowledge, this paper has shown how it is possible to estimate the regularization parameter and the unknown function via a single average consensus operation.

Overall, the theoretical achievements and the numerical strategies here described provide sound tools to reconstruct static functions in distributed networks. An important future direction is to extend all the analysis to an even more challenging situation where the unknown map may change in time and has to be tracked in an on-line manner.

\appendix

\subsection{Preliminary results}
\label{sec:ancillary_results_on_the_computation_of_conditional_expectations}

The following result will be especially useful in what follows.  In fact, it will be often used to obtain bounds on intricate conditional expectations just calculating unconditional means.

\begin{lemma}
	Let $\Omega$ denote a sample space, $\omega \in \Omega$ its generic element. Let $\mathcal{E}$ be an event such that 
	\begin{equation}
		\ProbabilityOf{ \omega \in \mathcal{E} } \geq 1 - \alpha.
		\label{equ:probability_of_omega_in_mathcal_E}
	\end{equation}

	If $g(\omega)$ is positive and~\eqref{equ:probability_of_omega_in_mathcal_E} holds then
	\begin{equation}
		\ExpectationOf{ g(\omega) \; \left| \; \omega \in \mathcal{E} \right. }
		\leq 
		\frac{1}{1 - \alpha} \ExpectationOf{ g(\omega) } .
		\label{equ:conditional_expectation_on_mathcal_E}
	\end{equation}
	\label{thm:conditional_expectation_on_mathcal_E}
\end{lemma}
\begin{proof}{Lemma~\ref{thm:conditional_expectation_on_mathcal_E}}
	Let $\eta$ be the probability measure on the $\sigma$-algebra $\Omega$ is equipped with. In general, for every $\mathcal{E}'$,
	\begin{equation}
		\begin{array}{ll}
		\ProbabilityOf{ \omega \in \mathcal{E}' \left| \omega \in \mathcal{E} \right. }
		& = \displaystyle
		\frac{\ProbabilityOf{ \omega \in \mathcal{E}' \cap \mathcal{E} }}{ \ProbabilityOf{ \omega \in \mathcal{E} }} \vspace{0.1cm} \\
		& \leq \displaystyle
		\frac{\ProbabilityOf{ \omega \in \mathcal{E}' }}{ \ProbabilityOf{ \omega \in \mathcal{E} }}
		\leq
		\frac{1}{1 - \alpha} \ProbabilityOf{ \omega \in \mathcal{E}' } .
		\end{array}
		\label{equ:minoration_of_mathcal_e_prime}
	\end{equation}
	If $\eta_{\mathcal{E}}$ denotes the probability measure $\eta$ conditional on $\mathcal{E}$, 
	one then has
	\begin{eqnarray}
		\int_{\mathcal{E}} g(\omega) d\eta_{\mathcal{E}}(\omega)
		& \leq &
		\frac{1}{1-\alpha} \int_{\mathcal{E}} g(\omega) d\eta(\omega) \\
		& \leq &
		\frac{1}{1-\alpha} \int_{\Omega} g(\omega) d\eta(\omega) . \\
	\end{eqnarray}
	\begin{flushright}
$\blacksquare$
\end{flushright}
\end{proof}

The following result exploits the Chernoff bound and will be important to obtain $\mathrm{Bnd}_{A}$ and $\mathrm{Bnd}_{B}$.
It will also clarify the role played by the $\varepsilon$ entering the bounds. 
\begin{lemma}
	Let $\alpha \in (0, 1)$ be a desired confidence level (e.g., $0.01$ or $0.05$), and $\varepsilon \in (0, 1]$ represent a given distance index for $\lambda_{\textrm{min}}$ and $\lambda_{\textrm{max}}$ as specified in~\eqref{equ:chernoff_bound_on_minimal_eigenvalue_of_G_T_G_second_version} and~\eqref{equ:chernoff_bound_on_maximal_eigenvalue_of_G_T_G_second_version}. If $E, M$ and $k$ in~\eqref{equ:eigenfunctions_are_uniformly_bounded} satisfy~\eqref{equ:condition_on_varepsilon_for_estimator_A} then
	\begin{equation}
		\ProbabilityOf
		{
			\lambda_{\textrm{min}} \left( \frac{G^{T} G}{M} \right)
			\geq
			\varepsilon 
		}
		\geq
		1 - \alpha .
		\label{equ:chernoff_bound_on_minimal_eigenvalue_of_G_T_G_second_version}
	\end{equation}
	If instead $E, M$ and $k$ satisfy~\eqref{equ:condition_on_varepsilon_for_estimator_B} then
	\begin{equation}
		\ProbabilityOf
		{
			\lambda_{\textrm{min}} \left( \frac{G^{T} G}{M} \right)
			\geq
			\varepsilon 
			\; \cap \;
			\lambda_{\textrm{max}} \left( \frac{G^{T} G}{M} \right)
			\leq
			2 - \varepsilon 
		}
		\geq
		1 - \alpha .
		\label{equ:chernoff_bound_on_maximal_eigenvalue_of_G_T_G_second_version}
	\end{equation}
	\label{thm:chernoff_bound_on_eigenvalues_of_G_T_G}
\end{lemma}

\begin{proof}{Lemma~\ref{thm:chernoff_bound_on_eigenvalues_of_G_T_G}} 
Since the assumptions in~\cite[Thm.~1.1]{tropp__2012__user_friendly_tail_bounds_for_sums_of_random_matrices} are satisfied, for any $\varepsilon \in (0, 1]$ one has
	\begin{equation}
		\ProbabilityOf
		{
			\lambda_{\min} \left( \frac{G^{T} G}{M} \right)
			\leq
			\varepsilon
		}
		\leq
		E
		\left(
			\frac{e^{-(1-\varepsilon)}}{\varepsilon^{\varepsilon}}
		\right)^{\displaystyle \frac{M}{E k}} .
	\label{equ:auxiliary_chernoff_bound_for_lambda_min}
	\end{equation}
	Condition \eqref{equ:condition_on_varepsilon_for_estimator_A} is obtained by picking $\alpha$ larger than the \ac{RHS} of~\eqref{equ:auxiliary_chernoff_bound_for_lambda_min} and manipulating this inequality. Then, \eqref{equ:chernoff_bound_on_minimal_eigenvalue_of_G_T_G_second_version} follows from \eqref{equ:auxiliary_chernoff_bound_for_lambda_min} just considering that if $\overline{\star}$ is the complementary of $\star$ then $\ProbabilityOf{\star} \leq \alpha \Leftrightarrow \ProbabilityOf{\overline{\star}} \geq 1 - \alpha$. Now, we can use again~\cite[Thm.~1.1]{tropp__2012__user_friendly_tail_bounds_for_sums_of_random_matrices} to claim that, for every $\varepsilon \in [0, 1]$, 
	\begin{equation}
		\ProbabilityOf
		{
			\lambda_{\max} \left( \frac{G^{T} G}{M} \right)
			\geq
			2 - \varepsilon
		}
		\leq
		E
		\left(
			\frac{e^{(1-\varepsilon)}}{(2-\varepsilon)^{(2-\varepsilon)}}
		\right)^{\displaystyle \frac{M}{E k}} .
	\label{equ:auxiliary_chernoff_bound_for_lambda_max}
	\end{equation}
	Let the arguments in the $\ProbabilityOf{\cdot}$ in the \ac{LHS} of~\eqref{equ:auxiliary_chernoff_bound_for_lambda_min} and~\eqref{equ:auxiliary_chernoff_bound_for_lambda_max} be respectively $\star_{\lambda_{\min}}$ and $\star_{\lambda_{\max}}$. Let also the \acp{RHS} of~\eqref{equ:auxiliary_chernoff_bound_for_lambda_min} and~\eqref{equ:auxiliary_chernoff_bound_for_lambda_max} be upper bounded respectively by $\alpha_{\lambda_{\min}}$ and $\alpha_{\lambda_{\max}}$. Then, it follows that
	\begin{equation}
		\begin{array}{ll}
		\ProbabilityOf{\star_{\lambda_{\min}} \cup \star_{\lambda_{\max}}}
		& \leq
		\ProbabilityOf{\star_{\lambda_{\min}}} + \ProbabilityOf{\star_{\lambda_{\max}}}
		\\
		& \leq
		\alpha_{\lambda_{\min}} + \alpha_{\lambda_{\max}}
		\leq
		2 \alpha_{\lambda_{\min}}
		\end{array}
	\end{equation}
	with the last inequality following from the fact that $\varepsilon \in \left[ 0, 1 \right] \implies \alpha_{\lambda_{\min}} \geq \alpha_{\lambda_{\max}}$ since
	\begin{equation}
		\frac{e^{-(1-\varepsilon)}}{\varepsilon^{\varepsilon}}
		\geq
		\frac{e^{(1-\varepsilon)}}{(2-\varepsilon)^{(2-\varepsilon)}} .
	\end{equation}
	Thus, letting the novel $\alpha$ be $2 \alpha_{\lambda_{\min}}$ (i.e., assuming~\eqref{equ:condition_on_varepsilon_for_estimator_B} to be satisfied), we obtain
	\begin{equation}
		\ProbabilityOf{\overline{\star}_{\min} \cap \overline{\star}_{\max}}
		\geq
		1 - \alpha ,
	\end{equation}
	and this proves \eqref{equ:chernoff_bound_on_maximal_eigenvalue_of_G_T_G_second_version}. \qed
	\begin{flushright}
$\blacksquare$
\end{flushright}
\end{proof}

The following lemma is just a generalization of the fact that convergence in mean ($L_1$-norm) of random variables implies 
convergence in probability. The proof is simple and therefore omitted.

\begin{lemma}
Let $g(\bm{x})$ denote a stochastic variable whose randomness derives from the input locations $\bm{x} \DefinedAs \left[ x_{1}, \ldots, x_{M} \right]^{T}$. Assume that $g(\bm{x}) \geq q$ almost surely with $q$ independent of $M$. In addition, assume also that for any $1-\alpha$ and $\varepsilon>0$ there exists $M_0$ such that $\forall M \geq M_0$ one has 
$$
\overline{g(\bm{x})} \leq q + \varepsilon \quad \text{with probability} \quad 1 - \alpha,
$$ 
in accordance with \Definition~\ref{def:BarE}. Then
	\begin{equation}
		\displaystyle
		\lim_{M \rightarrow + \infty}
		g(\bm{x})
		=
		q
		\qquad
		\text{in probability.}
	\end{equation}
	\label{Lemmaq}
		\begin{flushright}
$\blacksquare$
\end{flushright}
\end{lemma}

\subsection{Proof of Theorem \ref{thm:bounds_on_errors_of_f_A_and_f_B}} 

We start by computing the general expression for $\mathrm{Err}_{A}(\bm{x})$ in~\eqref{equ:definition_of_Err_A_and_Err_B}, then evaluate its expectation like in~\eqref{equ:definition-of-overline-Err-A}, and finally transport the results to the case of $\mathrm{Err}_{B}(\bm{x})$.

As for finding the general expression for $\mathrm{Err}_{A}(\bm{x})$, we recall the decomposition of the estimand as $f = f_{a} + f_{b}$ in~\eqref{equ:estimand_model}, the definition of $\mathcal{S}$ in~\eqref{equ:definition_of_mathcal_S} and the design requirement $\widehat{f}_{A}, \widehat{f}_{B} \in \mathcal{S}$, that imply $f_{a}, \widehat{f}_{A} \in \mathcal{S}$ and $f_{b} \in \mathcal{S}^{\perp}$. By construction, then, $\left\| f \right\|^{2} = \left\| f_{a} \right\|^{2} + \left\| f_{b} \right\|^{2}$ and
\begin{equation}
	\ExpectationOfGiven{ \left\| f - \widehat{f}_{A} \right\|^{2} }{\bm{x}}
	=
	\ExpectationOfGiven{ \left\| f_{a} - \widehat{f}_{A} \right\|^{2} }{\bm{x}}
	+
	\left\| f_{b} \right\|^{2}
	\label{equ:decomposition_of_MSE_for_estimator_A}
\end{equation}
where the expectations are w.r.t.\ the noises $\bm{\nu}$, so that $\ExpectationOfGiven{\left\| f_{b} \right\|^{2}}{\bm{x}} = \left\| f_{b} \right\|^{2}$ since 
$\bm{\nu},f_{b}$ and $\bm{x}$ are all mutually independent. Notice that a similar decomposition holds also for $\widehat{f}_{B}$.

As for $\ExpectationOfGiven{ \left\| f_{a} - \widehat{f}_{A} \right\|^{2} }{\bm{x}}$ in~\eqref{equ:decomposition_of_MSE_for_estimator_A}, we notice that $\left\| \widehat{f}_{A} \right\|^{2} = \left\| \widehat{\bm{a}} \right\|^{2}_{2} = \left\| H_{A} \bm{y} \right\|^{2}_{2}$. Since~\eqref{equ:measurement_model_vector_form} implies
\begin{equation}
	\widehat{\bm{a}}
	=
	H_{A}
	\left(
		G \bm{a} 
		+
		Z \bm{b}
		+
		\bm{\nu}
	\right) ,
\end{equation}
and since both $\bm{a} \perp \bm{b}$ and $\bm{\nu} \perp \bm{b}$, it eventually follows that
\begin{equation}
	\begin{array}{ll}
		\mathrm{Err}_{A}(\bm{x})
		= & \phantom{+}
		 \ExpectationOfGiven{ \left\| \bm{a} - H_{A} ( G \bm{a} + \bm{\nu} ) \right\|^{2}_{2} }{\bm{x}} \\ & \vspace{0.1cm}
		+ \ExpectationOfGiven{ \left\| H_{A} Z \bm{b} \right\|^{2}_{2} }{\bm{x}} \\ &
		+  \left\| f_{b} \right\|^{2}.
	\end{array}
	\label{equ:summary_of_MSE_of_estimator_A}
\end{equation}

\begin{proof}{equations~\eqref{equ:definition_of_bound_A} and~\eqref{equ:value_of_bound_A}}

Let $\overline{\mathcal{E}}$ be the event
\begin{equation}
	\overline{\mathcal{E}}
	\DefinedAs
	\left\{
		\lambda_{\textrm{min}} \left( \frac{G^{T} G}{M} \right)
		\geq
		\varepsilon
	\right\} ,
	\label{equ:definition_of_the_event_mathcal_E}
\end{equation}
and assume $\varepsilon$, $\alpha$, $M$ and $E$ satisfy~\eqref{equ:condition_on_varepsilon_for_estimator_A}. Since in this case we can apply \Theorem~\ref{thm:chernoff_bound_on_eigenvalues_of_G_T_G}, it holds that $\ProbabilityOf{\, \overline{\mathcal{E}} \,} \geq 1 - \alpha$.

We can now write the \ac{LHS} of~\eqref{equ:definition-of-overline-Err-A} as
\begin{equation}
	\ExpectationOfGiven
	{\mathrm{Err}_{A}(\bm{x})}
	{\overline{\mathcal{E}}}
	=
	\ExpectationOfGiven
	{
		\ExpectationOfGiven{ \left\| f - \widehat{f}_{A} \right\|^{2} }{\bm{x}}
	}
	{\overline{\mathcal{E}}} .
	\label{equ:double-expectation}
\end{equation}

Since $f_{b} \perp \overline{\mathcal{E}}$,~\eqref{equ:summary_of_MSE_of_estimator_A} implies
\begin{equation}
	\begin{array}{l}
		\ExpectationOfGiven{\mathrm{Err}_{A}(\bm{x})}{\overline{\mathcal{E}}} \vspace{0.1cm} 
		= \\ \quad \phantom{+} \vspace{0.1cm} 
		\ExpectationOfGiven
		{
			\ExpectationOfGiven{\left\| \bm{a} - H_{A} ( G \bm{a} + \bm{\nu} ) \right\|^{2}}{\bm{x}}
		}
		{\overline{\mathcal{E}}} \\ \vspace{0.1cm}
		\quad +
		\ExpectationOfGiven
		{
			\ExpectationOfGiven{\left\| H_{A} Z \bm{b} \right\|^{2}}{\bm{x}}
		}
		{\overline{\mathcal{E}}}
		\\
		\quad + \ExpectationOf{ \left\| f_{b} \right\|^{2} }.
	\end{array}
	\label{equ:summary_of_MSE_of_estimator_A_conditioned_on_mathcal_E}
\end{equation}

As for $\ExpectationOf{ \left\| f_{b} \right\|^{2} }$, we know from~\eqref{equ:estimand_model}, \eqref{equ:coefficients_are_gaussianly_distributed:b} and the mutual independence of the $b_{e}$'s, that
\begin{equation}
	\ExpectationOf{ \left\| f_{b} \right\|^{2} } = \sum_{e = E + 1}^{+\infty} \lambda_{e} .
	\label{equ:third_MSE_term_for_estimator_A}
\end{equation}
This term is thus an approximation error influenced only by the dimension $E$ of our search space $\mathcal{S}$.

Given~\eqref{equ:third_MSE_term_for_estimator_A}, what we actually need to bound is the first two terms in the \ac{RHS} of~\eqref{equ:summary_of_MSE_of_estimator_A_conditioned_on_mathcal_E}. We perform this task in the next two subsections.

\end{proof}

\subsection*{Bounding $\ExpectationOfGiven{\ExpectationOfGiven{\left\| H_{A} Z \bm{b} \right\|^{2}}{\bm{x}}}{\overline{\mathcal{E}}}$ in~\eqref{equ:summary_of_MSE_of_estimator_A_conditioned_on_mathcal_E}}
\label{ssec:characterization_of_first_term_of_f_A_conditioned_on_mathcal_E}

Exploiting the nature of the event $\overline{\mathcal{E}}$ to bound $H_{A}$ in~\eqref{equ:definition_of_H_A}, it is not difficult to prove that
\begin{equation} \footnotesize
	\ExpectationOfGiven{\ExpectationOfGiven{\left\| H_{A} Z \bm{b} \right\|^{2}}{\bm{x}}}{\overline{\mathcal{E}}}
	\leq
	\ExpectationOfGiven
	{
		\left\|
		\left( \varepsilon I_{E} + \frac{\sigma^{2}_{\nu}}{M} \Lambda^{-1}_{E} \right)^{-1}
			\frac{G^{T} Z}{M} \bm{b} \,
		\right\|^{2}
		\hspace{-0.2cm}
	}
	{\overline{\mathcal{E}}} .
	\label{equ:decomposition_of_second_term_of_the_error_for_estimator_A}
\end{equation}
Defining
\begin{equation}
	d_{e}
	\DefinedAs
	\frac{ \varepsilon M \lambda_{e} + \sigma^{2}_{\nu} }{ M \lambda_{e} },
	\qquad
	e = 1, \ldots, E,
	\label{equ:definition_of_d_e}
\end{equation}
it follows that 
\begin{equation}
	\left( \varepsilon I_{E} + \frac{\sigma^{2}_{\nu}}{M} \Lambda^{-1}_{E} \right)^{-1}
	=
	\DiagonalMatrixOf{ d_{1}^{-1}, \ldots, d_{E}^{-1} } .
	\label{equ:rewriting_of_varespilonIplusthings}
\end{equation}
Consider moreover that from the definition of $f_{b}$ in~\eqref{equ:estimand_model}, of $\bm{b}$ in~\eqref{equ:definition_of_bm_a_and_bm_b} and of $Z$ in~\eqref{equ:definition_of_Z} it follows that $\left[ Z \bm{b} \right]_{m} = f_{b} \left( x_{m} \right)$. Let then
\begin{equation}
	c_{e}
	\DefinedAs
	\left[ G^{T} Z \bm{b} \right]_{e}
	=
	\sum_{m = 1}^{M} \phi_{e} \left( x_{m} \right) f_{b} \left( x_{m} \right)
	\qquad
	e = 1, \ldots, E
	\label{equ:definition_of_c_e}
\end{equation}
so that
\begin{equation}
	\bm{b}^{T} Z^{T} G\left( \varepsilon I_{E} + \frac{\sigma^{2}_{\nu}}{M} \Lambda^{-1}_{E} \right)^{-2} G^{T} Z \bm{b} 
	=
	\sum_{e = 1}^{E}
	\frac{c_{e}^{2}}{d_e^2} \; .
	\label{equ:rewriting_of_bHGGHb}
\end{equation}
Combining~\eqref{equ:rewriting_of_varespilonIplusthings} and~\eqref{equ:rewriting_of_bHGGHb}, and considering that the $d_{e}$'s are deterministic, we can thus rewrite~\eqref{equ:decomposition_of_second_term_of_the_error_for_estimator_A} as
\begin{equation}
	\begin{array}{rl}
		\displaystyle
		\ExpectationOfGiven{\ExpectationOfGiven{\left\| H_{A} Z \bm{b} \right\|^{2}}{\bm{x}}}{\overline{\mathcal{E}}}
		& \leq \displaystyle
		\frac{1}{M^{2}}
		\sum_{e = 1}^{E}
		\frac{ \ExpectationOfGiven{ c_{e}^{2} }{\overline{\mathcal{E}}}}{d_{e}^{2}} \\
		& \leq \displaystyle
		\frac{1}{(1 - \alpha) M^{2}}
		\sum_{e = 1}^{E}
		\frac{ \ExpectationOf{ c_{e}^{2} }}{d_{e}^{2}} ,
	\end{array}
	\label{equ:rewriting_of_decomposition_of_second_term_of_the_error_for_estimator_A_in_expectation}
\end{equation}
where in the last inequality we applied \Lemma~\ref{thm:conditional_expectation_on_mathcal_E}. In view of the definition of the $c_{e}$'s in~\eqref{equ:definition_of_c_e} and the linearity of $\ExpectationOf{\cdot}$, one has
\begin{equation}
	\begin{array}{ll}
		\ExpectationOf{ c_{e}^{2} }
		= & \displaystyle \phantom{+}
		\sum_{m = 1}^{M}
		\ExpectationOf
		{
			\phi_{e}^{2} \left( x_{m} \right)
			f_{b}^{2} \left( x_{m} \right)
		} \\
		& \displaystyle +
		\sum_{m \neq m'}
		\ExpectationOf
		{
			\phi_{e} \left( x_{m} \right)
			\phi_{e} \left( x_{m'} \right)
			f_{b} \left( x_{m} \right) 
			f_{b} \left( x_{m'} \right) 
		} .
	\end{array}
	\label{equ:bound_on_expected_c_e_2}
\end{equation}
As for the first term in the \ac{RHS} of~\eqref{equ:bound_on_expected_c_e_2}, combining~\eqref{equ:coefficients_are_gaussianly_distributed:b} with bound~\eqref{equ:eigenfunctions_are_uniformly_bounded},
one immediately has
\begin{equation}
	\ExpectationOf
	{
		\phi_{e}^{2} \left( x_{m} \right)
		f_{b}^{2} \left( x_{m} \right) 
	}
	\leq
	k \sum_{e = E + 1}^{+ \infty} \lambda_{e} .
	\label{equ:bound_on_expectation_of_phi_e_2_f_b_2}
\end{equation}
As for the second term in the \ac{RHS} of~\eqref{equ:bound_on_expected_c_e_2}, due to the independence of the $\left\{ x_{m} \right\}_{m = 1}^{M}$ we know that
\begin{equation}
	\begin{array}{l}
		\ExpectationOfGiven{\phi_{e} \left( x_{m} \right)
			\phi_{e} \left( x_{m'} \right)
			f_{b} \left( x_{m} \right) 
			f_{b} \left( x_{m'} \right) }{f_{b}}
		= \\ \quad \displaystyle
		\ExpectationOfGiven{\phi_{e} \left( x_{m} \right)
			f_{b} \left( x_{m} \right) }{f_{b}}
		\ExpectationOfGiven{\phi_{e} \left( x_{m'} \right)
			f_{b} \left( x_{m'} \right)}{f_{b}}.
	\end{array}
\end{equation}
Moreover, since $e = 1, \ldots, E$, from the definition of $f_{b}$ in~\eqref{equ:estimand_model} it comes that 
\begin{equation}
	\begin{array}{l}
	\ExpectationOfGiven{ \phi_{e} \left( x_{m} \right) f_{b} \left( x_{m} \right) }{f_{b}} = 0.
	\end{array}
\end{equation}
Combining the two results, one has
\begin{equation}
	\ExpectationOf{ c_{e}^{2} }
	\leq
	k M \sum_{e = E + 1}^{+ \infty} \lambda_{e}
	\quad e = 1, \ldots, E .
	\label{equ:expectation-of-c-e-2-equals-k-m-blabla}
\end{equation}

Finally, combining~\eqref{equ:definition_of_d_e}, \eqref{equ:rewriting_of_decomposition_of_second_term_of_the_error_for_estimator_A_in_expectation}, \eqref{equ:expectation-of-c-e-2-equals-k-m-blabla} and \Lemma~\ref{thm:conditional_expectation_on_mathcal_E} one obtains
\vspace{0.2cm}
\begin{equation}
	\begin{array}{l}
		\displaystyle
		\ExpectationOfGiven{\ExpectationOfGiven{\left\| H_{A} Z \bm{b} \right\|^{2}}{\bm{x}}}{\overline{\mathcal{E}}} \leq \\
		\qquad \leq \displaystyle
		\frac{k M}{1-\alpha}
		\left(
			\sum_{e = 1}^{E}
			\frac{\lambda^2_{e}}{(\varepsilon M \lambda_e + \sigma^{2}_{\nu})^2}
		\right)
		\left(
			\sum_{e = E + 1}^{+\infty} \lambda_{e}
		\right) .
	\end{array}
	\label{equ:second_MSE_term_for_estimator_A}
\end{equation}

\subsection*{Bounding $\ExpectationOfGiven{ \ExpectationOfGiven{ \left\| \bm{a} - H_{A} ( G \bm{a} + \bm{\nu} ) \right\|^{2} }{\bm{x}}}{\overline{\mathcal{E}}}$ in~\eqref{equ:summary_of_MSE_of_estimator_A_conditioned_on_mathcal_E}}
\label{ssec:characterization_of_second_term_of_f_A_conditioned_on_mathcal_E}

To characterize $\left\| \bm{a} - H_{A} ( G \bm{a} + \bm{\nu} ) \right\|^{2}$, note that this term corresponds to the \ac{MSE} of a classical \ac{MAP} estimator for a standard linear and finite-dimensional Gaussian model (where the term $\bm{b}$ is not involved). More precisely, if the measurements models conditional on $\bm{x}$ were
\begin{equation}
\bm{y}= G\bm{a} +\nu, \quad \bm{a} \sim\mathcal{N}(0,\Lambda), \quad v \sim\mathcal{N}(0,\sigma_{\nu}^2)
	\label{StandardModelA}
\end{equation}
with $\bm{a}$ independent of $\nu$, the optimal estimator would indeed be
\begin{equation}
	\widehat{a}_A = H_A \bm{y} = H_A( G \bm{a} + \bm{\nu} ).
	\label{StandardEstA}
\end{equation}
Exploiting standard results on Gaussian estimation, the covariance matrix of the error is
\begin{equation}
	\VarianceOfGiven{ \bm{a} - H_{A} ( G \bm{a} + \bm{\nu} )}{\bm{x} }
	=
	\frac{\sigma^{2}_{\nu}}{M}
	\left( \frac{G^{T} G}{M} + \frac{\sigma^{2}_{\nu}}{M} \Lambda^{-1}_{E} \right)^{-1} .
	\label{StandardVar}
\end{equation}
Applying \Lemma~\ref{thm:conditional_expectation_on_mathcal_E} and using~\eqref{equ:definition_of_the_event_mathcal_E} to bound $H_{A}$ in~\eqref{equ:definition_of_H_A}, from definitions~\eqref{equ:definition_of_d_e} and~\eqref{equ:rewriting_of_varespilonIplusthings} it follows that
\begin{equation} \footnotesize
	\ExpectationOfGiven{ \ExpectationOfGiven{ \left\| \bm{a} - H_{A} ( G \bm{a} + \bm{\nu} ) \right\|^{2} }{\bm{x}}}{\overline{\mathcal{E}}}
	\leq
	\frac{\sigma^{2}_{\nu}}{1-\alpha}
	\left(
		\sum_{e = 1}^{E}
		\frac{\lambda_e}{\varepsilon M \lambda_e + \sigma^{2}_{\nu}}
	\right) .
	\label{equ:first_MSE_term_for_estimator_A}
\end{equation}
Hence, the bound on $\mathrm{Err}_{A}(\bm{x})$ is obtained by combining~\eqref{equ:third_MSE_term_for_estimator_A},~\eqref{equ:second_MSE_term_for_estimator_A} and \eqref{equ:first_MSE_term_for_estimator_A}. \hfill $\blacksquare$

\begin{proof}{equations~\eqref{equ:definition_of_bound_B} and~\eqref{equ:value_of_bound_B}}

Let $\overline{\mathcal{E}}$ be now the event
\begin{equation}
	\overline{\mathcal{E}}
	\DefinedAs
	\left\{
		\lambda_{\textrm{min}} \left( \frac{G^{T} G}{M} \right)
		\geq
		\varepsilon
		\cap
		\lambda_{\textrm{max}} \left( \frac{G^{T} G}{M} \right)
		\leq
		2 - \varepsilon
	\right\} ,
	\label{equ:definition_of_the_event_mathcal_E_for_max}
\end{equation}
and assume that $\varepsilon$, $\alpha$, $M$ and $E$ satisfy~\eqref{equ:condition_on_varepsilon_for_estimator_B}. Substituting $H_{A}$ with $H_{B}$ in the derivation of~\eqref{equ:summary_of_MSE_of_estimator_A_conditioned_on_mathcal_E} one obtains
\begin{equation}
	\begin{array}{ll}
		\ExpectationOfGiven{\mathrm{Err}_{B}(\bm{x})}{\overline{\mathcal{E}}}  \\
		& = 
		\ExpectationOfGiven{ \ExpectationOfGiven{ \left\| \bm{a} - H_{B} ( G \bm{a} + \bm{\nu} ) \right\|^{2} }{\bm{x}}}{\overline{\mathcal{E}}} \\ & \vspace{0.1cm}
		+ \ExpectationOfGiven{ \ExpectationOfGiven{\left\| H_{B} Z \bm{b} \right\|^{2}}{\bm{x}} }{\overline{\mathcal{E}}}\\ &
		+ \ExpectationOf{ \left\| f_{b} \right\|^{2} }
	\end{array}
	\label{equ:summary_of_MSE_of_estimator_B_conditioned_on_mathcal_E}
\end{equation}
We already know that $\ExpectationOf{ \left\| f_{b} \right\|^{2} } = \sum_{e = E + 1}^{+\infty} \lambda_{e} $. Hence, we have to bound the first two terms in the \ac{RHS} of~\eqref{equ:summary_of_MSE_of_estimator_B_conditioned_on_mathcal_E}.

\end{proof}

\subsection*{Bounding $\ExpectationOfGiven{ \ExpectationOfGiven{\left\| H_{B} Z \bm{b} \right\|^{2}}{\bm{x}} }{\overline{\mathcal{E}}}$ in~\eqref{equ:summary_of_MSE_of_estimator_B_conditioned_on_mathcal_E}}
\label{ssec:characterization_of_first_term_of_f_B_conditioned_on_mathcal_E}

From the definition of $H_{B}$ in~\eqref{equ:definition_of_H_B}, one has
\begin{equation}\footnotesize
	\ExpectationOfGiven{ \ExpectationOfGiven{\left\| H_{B} Z \bm{b} \right\|^{2}}{\bm{x}} }{\overline{\mathcal{E}}}
	\leq
	\ExpectationOfGiven
	{
		\left\|
		\left( I_{E} + \frac{\sigma^{2}_{\nu}}{M} \Lambda^{-1}_{E} \right)^{-1}
			\frac{G^{T} Z}{M} \bm{b} \,
		\right\|^{2}
		\hspace{-0.2cm}
	}
	{\overline{\mathcal{E}}} ,
	\label{equ:decomposition_of_second_term_of_the_error_for_estimator_B}
\end{equation}
which corresponds to \eqref{equ:decomposition_of_second_term_of_the_error_for_estimator_A} with $\varepsilon = 1$. Thus, we just need to plug $\varepsilon = 1$ in \eqref{equ:second_MSE_term_for_estimator_A} to obtain the desired result, i.e.,
\begin{equation}
	\begin{array}{l}
		\displaystyle
		\ExpectationOfGiven{ \ExpectationOfGiven{\left\| H_{B} Z \bm{b} \right\|^{2}}{\bm{x}} }{\overline{\mathcal{E}}} \leq \\
		\qquad \leq \displaystyle
		\frac{k M}{1 - \alpha}
		\left(
			\sum_{e = 1}^{E} \frac{\lambda_{e}^{2}}{\left( M \lambda_{e} + \sigma^{2}_{\nu} \right)^{2}}
		\right)
		\left( \sum_{e = E + 1}^{+\infty} \lambda_{e} \right) .
	\end{array}
	\label{equ:second_MSE_term_for_estimator_B}
\end{equation}

\subsection*{Bounding $\ExpectationOfGiven{ \ExpectationOfGiven{ \left\| \bm{a} - H_{B} ( G \bm{a} + \bm{\nu} ) \right\|^{2}}{\bm{x}} }{\overline{\mathcal{E}}}$ in~\eqref{equ:summary_of_MSE_of_estimator_B_conditioned_on_mathcal_E}}
\label{ssec:characterization_of_second_term_of_f_B_conditioned_on_mathcal_E}

It is useful again to reason as if the measurements were generated according to~\eqref{StandardModelA} so that $\bm{y}=G \bm{a} + \bm{\nu}$. Then, let 
$$
\widehat a_B = H_B \bm{y}
$$
and 
$$
\Phi_B= \VarianceOf{\bm{a}-\widehat a_B | \bm{x}}.
$$ 
Recalling \eqref{StandardEstA} and \eqref{StandardVar}, let also
$$
\widehat{a}_A = H_A \bm{y}
$$
and 
$$
\Phi_A= \VarianceOf{\bm{a}-\widehat a_A | \bm{x}}=\frac{\sigma^{2}_{\nu}}{M}
	\left( \frac{G^{T} G}{M} + \frac{\sigma^{2}_{\nu}}{M} \Lambda^{-1}_{E} \right)^{-1}.
$$ 
After some simple calculations one obtains
\begin{equation}
\Phi_B = \Phi_A 
	+  \underbrace{(H_A-H_B)(G\Lambda G^T+\sigma_{\nu}^2 I)(H_A-H_B)^T}_{\widetilde{\Phi}}.
\end{equation}
The bound for $\ExpectationOfGiven{\TraceOf{\Phi_A}}{\overline{\mathcal{E}}}$ was already obtained in \eqref{equ:first_MSE_term_for_estimator_A} so that now we can just focus on bounding $\ExpectationOfGiven{\TraceOf{\widetilde{\Phi}}}{\overline{\mathcal{E}}}$. Define
$$
A=\frac{G^TG}{M}+\frac{\sigma_{\nu}^{2}\Lambda^{-1}}{M}, \quad B=I+\frac{\sigma_{\nu}^{2}\Lambda^{-1}}{M}.
\vspace{0.2cm} 
$$
Then, it follows that
$$
\begin{array}{lll}
	\widetilde{\Phi} \!\!\! &= (A^{-1}-B^{-1}) \left( \frac{G^TG}{M}\Lambda\frac{G^TG}{M} + \frac{G^TG}{M^2}\sigma_{\nu}^2 \right) (A^{-1}-B^{-1}) \\
	&= A^{-1}(\underbrace{B-A}_{\IDefinedAs C})B^{-1} \Big(\frac{G^TG}{M}\Lambda\frac{G^TG}{M} + \frac{G^TG}{M^2}\sigma_{\nu}^2 \Big) \\
	 & \times B^{-1}(B-A)A^{-1} 
\end{array}
$$
so that
\begin{equation}
	\begin{array}{ll}
		\TraceOf{\widetilde{\Phi}} &=
		\TraceOf{B^{-1}CA^{-2}CB^{-1}\Big(\frac{G^TG}{M}\Lambda\frac{G^TG}{M} + \frac{G^TG}{M^2}\sigma_{\nu}^2 \Big)}\\
		&= \TraceOf{\Lambda^{1/2}\frac{G^TG}{M}B^{-1}CA^{-2}CB^{-1} \frac{G^TG}{M}\Lambda^{1/2} }  \\
		&\phantom{=} + \TraceOf{\sigma_{\nu}^2 A^{-1}CB^{-1} \frac{G^TG}{M^2} B^{-1}CA^{-1}}.
\end{array}
\end{equation}
Since
$$
	\varepsilon I \leq \frac{G^TG}{M}\leq (2-\varepsilon) I, \ \ \Lambda\leq\lambda_1I
$$
we obtain
$$
	\begin{array}{lll}
		A &\geq & \left(\varepsilon + \frac{\lambda_1^{-1} \sigma_{\nu}^2}{M} \right)I, \\
		B &\geq & \left(1 + \frac{\lambda_1^{-1} \sigma_{\nu}^2}{M} \right)I \geq  \left(\varepsilon + \frac{\lambda_1^{-1} \sigma_{\nu}^2}{M} \right)I,\\
		C^2 & = & \left(I-\frac{G^TG}{M} \right)^2\leq (1-\varepsilon)^2 I.
	\end{array}
$$
Exploiting such inequalities and \Lemma~\ref{thm:conditional_expectation_on_mathcal_E} we obtain    
$$
\begin{array}{l}
	\ExpectationOfGiven{\TraceOf{\widetilde{\Phi}}}{\overline{\mathcal{E}}}
	\leq \\
	\displaystyle
	\qquad 
	\left(\varepsilon + \frac{\lambda_1^{-1} \sigma_{\nu}^2}{M} \right)^{-4}
	(1-\varepsilon)^2
	\frac{(2-\varepsilon)^2}{1-\alpha}
	\left(\sigma_{\nu}^2\frac{E}{M} + \sum_{e=1}^E \lambda_e \right)
\end{array}
$$
which, combined with \eqref{equ:first_MSE_term_for_estimator_A}, provides the desired result, also leading to the overall bound for $\mathrm{Err}_{B}(\bm{x})$.

\subsection{Proof of Theorem \ref{thm:f_A_and_f_B_are_asymptotically_efficient}} {}

For both the cases the proof is obtained using Lemma \ref{Lemmaq}. In particular, for both the estimators the lower bound is 
$$
q=\sum_{e = E + 1}^{+\infty} \lambda_{e}
$$
according to Theorem \ref{thm:lower_bound_for_generic_estimator_of_f}. Then, from the expressions of the bounds \eqref{equ:value_of_bound_A} and \eqref{equ:value_of_bound_B}, it is immediate to verify that, for fixed $E$ and $\alpha$, they are monotonically decreasing in $M$. Now, for any $0<\delta<1$, let $\alpha=\delta$ and $\varepsilon$ such that $\kappa<\frac{\delta}{2}\frac{4\delta}{1-\alpha}\sum_{e=1}^\infty \lambda_e$. Then, there exists $M_0$ such that for $M>M_0$ conditions \eqref{equ:condition_on_varepsilon_for_estimator_A} and \eqref{equ:condition_on_varepsilon_for_estimator_B} are satisfied and 
$$
\overline{\mathrm{Err}}_{A} \leq q + \delta \quad \text{with probability} \quad 1 - \delta
$$ 
and
$$
\overline{\mathrm{Err}}_{B} \leq q + \delta \quad \text{with probability} \quad 1 - \delta.
$$
As anticipated, the use of Lemma \ref{Lemmaq} then concludes the proof.

\subsection{Proof of Theorem \ref{thm:f_A_consistent}} 

For what regards \eqref{equ:f_A_is_consistent}, it is sufficient to recall that, for finite $M$, as $E \rightarrow +\infty$ the estimator $\widehat{f}_{A}$ coincides with the \ac{MAP} estimator which is consistent, see \cite{pillonetto_bell__2007__bayes_and_empirical_bayes_semi-blind_deconvolution_using_eigenfunctions_of_a_prior_covariance}[Appendix] for details.

\subsection{Proof of Theorem \ref{thm:f_A_and_f_B_are_consistent}} 

The convergence \eqref{equ:f_B_is_consistent} related to $\widehat{f}_{B}$ is more delicate and we will exploit bound~\eqref{equ:definition_of_bound_B}. The rationale is to establish conditions on the convergence of $E,M$ to infinity to make both the confidence level $\alpha$ and the bound~\eqref{equ:value_of_bound_B} tend to zero. As for $\alpha$, we can make it vanish to zero by setting $\alpha = \frac{1}{E^{1-\delta}}$ with $\delta\in(0,1)$ . As for the bound~\eqref{equ:value_of_bound_B}, its first components naturally vanish with $E \rightarrow +\infty$, while the last two do not. In particular, one needs $\kappa$ in~\eqref{equ:definition_of_kappa} to go to zero and this requires $\varepsilon \rightarrow 1$. Hence, we set $\varepsilon$ such that $\frac{(1-\varepsilon)^2}{2k} = \frac{1}{M^{1-\delta}}$. However, the convergence of $\alpha$ and \eqref{equ:value_of_bound_B} is not enough, since condition~\eqref{equ:condition_on_varepsilon_for_estimator_B} must be always satisfied. This condition is indeed verified since
$$
	\frac{Ek}{M}\log\frac{E}{\alpha}
	=
	\frac{Ek}{M}\log E^\delta\leq \frac{kM^\delta}{M}
	=
	\frac{(1-\varepsilon)^2}{2}\leq 1-\varepsilon+\varepsilon\log\varepsilon.
	\qed
$$

\addcontentsline	{toc}{section}{References}

\end{document}